\documentclass[11pt]{amsart}
\usepackage{amssymb,amscd,amsthm,hyperref,cleveref}
\usepackage[shortlabels]{enumitem}
\usepackage[margin=0.8in]{geometry}

\usepackage{colortbl}
\definecolor{lightgray}{rgb}{0.9,0.9,0.9}

\def\co{\colon\thinspace}

\newcommand\PP{{\mathbb{P}}}
\newcommand\EE{\mathbb{E}}
\newcommand\cV{{\mathcal{V}}}
\newcommand\cA{{\mathcal{A}}}

\newcommand\ZZ{\mathbb{Z}}

\let\div\Div

\DeclareMathOperator{\Ber}{Ber}

\DeclareMathOperator{\KS}{KS}

\DeclareMathOperator{\Pic}{Pic}
\DeclareMathOperator{\Proj}{Proj}
\DeclareMathOperator{\red}{red}
\DeclareMathOperator{\Res}{Res}

\DeclareMathOperator{\NS}{NS}

\DeclareMathOperator{\id}{id}

\newcommand{\CC}{\mathbb{C}}
\newcommand{\cMn}{\cM_{0,n,0}}
\newcommand{\cCn}{\cC_{0,n,0}}

\newcommand{\cC}{\mathcal{C}}
\newcommand{\cL}{\mathcal{L}}
\newcommand{\cM}{\mathcal{M}}
\newcommand{\cD}{\mathcal{D}}
\newcommand{\del}{\partial}

\newcommand{\e}{\mathfrak{e}}
\newcommand{\eps}{\varepsilon}

\newcommand{\cF}{\mathcal{F}}
\newcommand{\cT}{\mathcal{T}}

\newcommand{\cO}{\mathcal{O}}

\newcommand{\mc}[1]{\mathcal{#1}}
\newcommand{\vol}{[dz | d\zeta]}

\newcommand{\sKS}{super Kodaira-Spencer }

\newtheorem{thm}{Theorem}

\newtheorem{prop}[thm]{Proposition}
\newtheorem{lem}[thm]{Lemma}

\theoremstyle{definition}

\theoremstyle{remark}

\newtheorem{rem}{Remark}     

\author[S. L. Cacciatori]{Sergio Luigi Cacciatori}
\address{Universit\`a dell’Insubria, dipartimento di scienza e alta tecnologia, via Valleggio 11, 22100, Como, Italy and INFN,
via Celoria 16, 20100, Milano, Italy}
\author[S. Grushevsky]{Samuel Grushevsky}
\address{Department of Mathematics and Simons Center for Geometry and Physics, Stony Brook University, Stony Brook, NY 11794-3651}
\author[A. A. Voronov]{Alexander A. Voronov}
\address{School of Mathematics,
  University of Minnesota,
  Minneapolis, MN 55455 and
 Kavli IPMU (WPI), UTIAS, The University of Tokyo, Kashiwa, Chiba 277-8583, Japan
 }
\thanks{Research of the second author is supported in part by NSF grant DMS-21-01631. Research of the third author is supported in part by World Premier International Research Center Initiative (WPI Initiative), MEXT, Japan, 
and Collaboration grant \#585720 from the Simons Foundation. }

\title
{Tree-Level Superstring Amplitudes: The Neveu-Schwarz Sector}

\date{February 22, 2024}

\begin{document}
\begin{abstract}
 We present a complete computation of superstring scattering amplitudes at tree level, for the case of Neveu-Schwarz insertions. Mathematically, this is to say that we determine explicitly the superstring measure on the moduli space $\cMn$ of super Riemann surfaces of genus zero with~$n \ge 3$ Neveu-Schwarz punctures. While, of course, an expression for the measure was previously known, we do this from first principles, using the canonically defined super Mumford isomorphism \cite{VorSuperMum}. We thus determine the scattering amplitudes, explicitly in the global coordinates on $\cMn$, without the need for picture changing operators or ghosts, and are also able to determine canonically the value of the coupling constant. Our computation should be viewed as a step towards performing similar analysis on $\cM_{0,0,n}$, to derive explicit tree-level scattering amplitudes with Ramond insertions.
\end{abstract}

\maketitle

\tableofcontents

\section{Introduction}
Bosonic string scattering amplitudes can be mathematically defined as integrals over the space of all maps from Riemann surfaces with marked points to a given spacetime manifold. To obtain an explicit expression for such an integrand, one proceeds from first principles of string theory, but of course eventually has to choose local coordinates on the moduli space.
This construction in particular relies on the \emph{Mumford isomorphism} $\det(\EE_1)^{\otimes 13} \cong \det(\EE_2)$ over the moduli space $M_g$ of genus $g > 1$ (bosonic) Riemann surfaces. Here $\EE_k\to M_g$ is the universal vector bundle of holomorphic $k$-differentials, whose fiber over a Riemann surface $X \in M_g$ is $H^0(X, \omega^{\otimes k})$, where $\omega = \cT^*_X$ is the holomorphic cotangent bundle of $X$. Algebro-geometrically, since $\EE_2$ is the cotangent bundle $\cT^*_{M_g}$, the Mumford isomorphism gives the expression $\omega_{M_g}=c_1( \cT^*_{M_g}) = \det (\EE_2) \cong \lambda_1^{\otimes 13}$ for the canonical bundle $\omega_{M_g}$ of the moduli space, where $\lambda_1:=c_1(\EE_1) = \det (\EE_1) \in \Pic (M_g)$. Since each space-time dimension contributes $\lambda_1^{\otimes 1/2}$ to the integrand, in particular the Mumford isomorphism implies that bosonic string theory could only make sense in 26 dimensions. We note that the bosonic Mumford isomorphism is a statement that two vector bundles are isomorphic, but it does not give a canonical isomorphism between them. 

Bosonic string amplitudes were extensively studied in the 1980s, and various expressions for them were obtained in arbitrary genus, see \cite{Belavin-Knizhnik,d'hoker-phong-86} and also the more recent literature \cite{Jost's book,Witten's notes} summarizing these results and providing further references. At tree level (i.e.~no loops, i.e.~genus zero Riemann surfaces) explicit expressions for bosonic string amplitudes are well-known, see \cite{polchinskibook1}.

\smallskip
The mathematical foundations of the superstring scattering amplitudes
were reviewed and revisited by Witten \cite{Witten}, with further
mathematical underpinnings developed in \cite{FKP1,FKP2,
  BruzzoPolishchukRuiperez}. Superstring scattering amplitudes can be
mathematically defined as integrals over the space of all maps from
super Riemann surfaces to a given space-time supermanifold. Recall that a \emph{super
Riemann surface $($SRS$)$} is a complex supermanifold~$X$ of complex
dimension~$(1|1)$, i.e.~such that local charts at any point of~$X$ are
given by one even and one odd coordinate, together with a
\emph{maximally non-integrable, odd distribution} $\cD\subset \cT_X$,
i.e.~a distribution $\cD$ of rank $(0|1)$ such that the Lie bracket of
super vector fields induces an isomorphism
\[
[-,-] \co \cD \otimes \cD \xrightarrow{\sim} \cT_X/\cD.
\]
This yields a canonical short exact sequence
\begin{equation}
\label{tangentsequence}
    0\to \cD\to \cT_X\to\cD^{\otimes 2}\to 0.
\end{equation}
The \emph{super
Mumford isomorphism} over the moduli space $\cM_g$ of SRS of genus $g
> 1$ is the isomorphism
\begin{equation}
\label{superMumford}
\Ber(\EE_{1/2})^{\otimes 5} = \Ber(\EE_{3/2})\,,
\end{equation}
where $\EE_{1/2}$ is now the virtual coherent sheaf over $\cM_g$, with
fiber $H^0(X, \omega) - H^1(X, \omega)$ over $X\in\cM_g$, where $\omega
:= \Ber \cT^*_X = \cD^*$ now denotes the {\em canonical bundle} of the
SRS $X$, while $\EE_{3/2}$ denotes the bundle with fibers $ H^0(X,
\omega^{\otimes 3})$ (for which the $H^1$ vanishes automatically), and
$\Ber$ denotes the Berezinian, also known as
superdeterminant.\footnote{Recall that on the bosonic Riemann surface
$X_{\red}$ underlying $X$, $\omega$ restricts to a square root of the
canonical bundle of $X_{\red}$, so it makes sense to talk of $1/2$-
and $3/2$-differentials, by taking this choice of the line bundle
$\omega_{X_{\red}}^{\otimes 1/2}$.} Since the cotangent bundle to $\cM_g$
is $\EE_{3/2}$, this gives
\begin{equation}
\label{volumeforms}
\Ber(\EE_{1/2})^{\otimes 5} = \omega_{\cM_g}\,,
\end{equation}
which again is the reason why superstring theory can only be consistent in $10=5\cdot 2$ dimensions (see below and \cite{Manin.1986.criticaldim, FKP1, FKP2, Witten} for more details). Crucially, unlike in the bosonic case, the super Mumford isomorphism \eqref{superMumford} is not just a statement that two line bundles over~$\cM_g$ are abstractly isomorphic, but gives a canonical isomorphism between them, see~\cite{VorSuperMum}.

\smallskip
While the moduli space $M_g$ of genus $g>1$ bosonic Riemann surfaces is a complex Deligne-Mumford stack of dimension~$3g-3$, the moduli space $\cM_g$ of genus $g>1$ SRS is a complex superstack of dimension $(3g-3\,|\,2g-2)$. For $g=0$ (resp.~1), to get a Deligne-Mumford stack as the moduli space of bosonic Riemann surfaces, one needs to consider Riemann surfaces together with at least $n = 3$ (resp., $n = 1$) marked points, whence $\dim M_{0,n}=n-3;\ \dim M_{1,n}=n$. Unlike the bosonic case, where there is only one notion of a marked point on a Riemann surface (which is a codimension $1$ complex subspace, and equivalent to considering a puncture), there are two differen notions of markings or punctures
for super Riemann surfaces. A Neveu-Schwarz (NS) marking on a super Riemann surface~$X$ is just a point of~$X$, i.e.~a subspace of complex dimension $(0|0)$ and codimension $(1|1)$, which plays no role for the distribution $\cD$. On the other hand, a Ramond (R) marking is a subspace of codimension $(1|0)$, over which the maximal non-integrability of the distribution $\cD$ fails in a controlled way \cite{Witten, Diroff, Ott-Voronov, Donagi-Ott}. 

We denote by $\cM_{g,n,k}$ the moduli of genus~$g$ super Riemann surfaces with $n$ Neveu-Schwarz markings and~$k$ Ramond markings, and then for $g>1$ it is a superstack of dimension $(3g-3+n+k\, |\, 2g-2+n+\tfrac{k}2)$. To write an explicit expression for superstring scattering amplitudes, one needs to choose local coordinates on $\cM_{g,n,k}$. We recall that locally any (smooth or complex) supermanifold is split, i.e.~is isomorphic to the exterior algebra of a vector bundle on its underlying (bosonic) manifold. Thus in a sufficiently small coordinate patch, $\cM_{g,n,k}$ is modelled on a vector bundle (quotiented by a finite group, as we are working with orbifolds) over the moduli space $SM_{g,n+k} = (\cM_{g,n,k})_{\red}$ of Riemann surfaces with a spin structure. However, as was shown by Donagi-Witten \cite{Donagi-Witten} (see also Donagi-Ott \cite{Donagi-Ott} for the case with Ramond markings), for high enough genus, globally $\cM_{g,n,k}$ is {\em not} split and does not even admit a global holomorphic projection to its reduced bosonic moduli space $SM_{g,n+k}$. In physics literature \cite{ListTextbooksOnStrings}, to write formulas on $\cM_{g,n,k}$, one often works in local coordinates, and then uses picture changing operators (PCOs) to change to a different coordinate chart. To deal with the choices involved, one introduces ghosts, and then integrates over them as well. In particular, this approach led to spectacular computations of superstring scattering amplitudes in genus 2 by D'Hoker and Phong \cite{d'hoker-phong-2000s}.

\smallskip
Our focus in this paper is the moduli space of genus~$0$ super Riemann
surfaces with only NS markings: $\dim \cMn=(n-3\,|\,n-2)$, for
$n\ge 3$, since the dimension of the automorphism group
$G:=\operatorname{OSp}(1|2)^0$ of the super projective line
$\PP^{1|1}$ with the odd distribution $\cD\subset\cT_{\PP^{1|1}}$ is equal
to $(3|2)$. Then~$\cMn$ is the moduli space of
$n$-tuples of labeled pairwise distinct points
$(z_1|\zeta_1),\dots,(z_n|\zeta_n)\in\PP^{1|1}$, called \emph{NS
punctures}, up to the diagonal action of $G$. The group $G$ acts
transitively on triples of distinct bosonic points, and then further
on pairs of odd coordinates. Moreover, fixing the values of
$z_1,z_2,z_3,\zeta_1,\zeta_2$ to be $z_1^0,z_2^0,z_3^0,\zeta_1^0,\zeta_2^0$ for some distinct values~$z_1^0,z_2^0,z_3^0\in\PP^1$
gives a point in each $G$-orbit in the configuration space of $n$-tuples of distinct
points of~$\PP^{1|1}$. Each $G$-orbit contains a unique element with these values, up to the action of the involution, which for the choice of $\zeta_1^0=\zeta_2^0=0$ is the element of $G$ that fixes the bosonic coordinate and  changes the sign of the odd coordinate on $\PP^{1|1}$, see \cite[\S\ 2.12]{Manin.1991.ting} and
\cite[\S\ 5.1.2]{Witten.2019.nosrsatm}. Thus, similarly to the bosonic
case, there are {\em global} coordinates $(z_4,\dots,z_n
\,|\,\zeta_3,\dots,\zeta_n)$ on $\cMn$, which exhibit this
supermoduli space as a dense open subset of the superorbifold
$\CC^{0|1}\times (\CC^{1|1})^{n-3}/\ZZ_2 \cong \CC^{n-3 \, | \,
  n-2}/\ZZ_2 $ (provided one of $z_1^0,z_2^0,z_3^0$ is set to~$\infty$), where $\ZZ_2$ acts
by the involution that changes the signs of all odd coordinates
simultaneously. Therefore it should be possible to write a global
expression for the tree-level superstring amplitudes, for the case of only NS punctures, explicitly in
these global coordinates, from first principles, without introducing
ghosts or picture changing operators. 

When we talk about integrals,
such as \eqref{eq:main}, over $\cMn$, we treat them as improper
integrals over a noncompact domain and ignore the isotropy group
$\ZZ_2$, since its only possible contribution to the integral could be a factor of~$2$. The choice of whether to divide or multiply by
2 is left to the future reader, as the additivity properties of the
Berezin integral over superstacks are not yet well understood.

The first principles we mentioned above are canonical isomorphisms,
such as the super Mumford isomorphism \eqref{superMumford}, which
simplifies in the case $g=0$. Over $\cMn$, the super vector bundle $\EE_{1/2}$ trivializes canonically, and thus and thus so does $\Ber (\EE_{1/2})$,
because $H^0 (\PP^{1|1}, \omega) = 0$ and $H^1(\PP^{1|1}, \omega) =
H^0 (\PP^{1|1}, \cO_{\PP^{1|1}})^*$ is canonically isomorphic to
$\CC$. However, in the presence of NS markings, the canonical bundle
$\omega_{\cM_{g,n,0}}$ is identified with $\Ber (H^0(\omega^{\otimes
  3}(\NS)))$, where $\NS$ is the total NS divisor, rather than with $\Ber
(\EE_{3/2})$, as in the case of no markings. With this change, the
super Mumford isomorphism \eqref{volumeforms} gets modified to a
canonical isomorphism 
\[
\Ber (\EE_{1/2})^{\otimes 5} \otimes \Ber (H^0(X, \omega^{\otimes 3}(\NS)|_{\NS})) = \omega_{\cM_{g,n,0}} \, ,
\]
see \cite{Witten, Diroff}. In the $g =0$ case, this simplifies to
\begin{equation}
\label{NS-Mumford}
\Ber (H^0(X, \omega^{\otimes 3}(\NS)|_{\NS})) = \omega_{\cMn} \, .
\end{equation}
It is known that the line bundle $\omega^{\otimes 3}(\NS)$ above may be replaced by any other odd line bundle via the canonical isomorphism \eqref{core}, \cite{VorSuperMum, Diroff}. As a consequence, we get the canonical isomorphism
\begin{equation}
\label{NS-Mumford1}
\Ber (H^0(X, \omega|_{\NS})) = \omega_{\cMn}\, ,
\end{equation}
which will be used in the main theorem below.

In this paper we perform the computation of the tree-level NS scattering amplitudes from first principles, not surprisingly recovering the known formulas \cite{Friedan-Martinec-Shenker1, Friedan-Martinec-Shenker2, Knizhnik, etc} from the physics literature. The main ingredient is the canonical determination of the holomorphic measure on $\cMn$, which we deduce by using the canonical super Mumford isomorphism; this also allows us to canonically fix the coupling constant. We show that the choice of left-moving, conformal vertex-operator insertions produces a section
\begin{equation}
\label{VO}
\left\langle
     \bigotimes_{j=1}^n V_j(z_j | \zeta_j) \Big[ \, \zeta [dz | d\zeta]\vert_{\NS_j} \Big\vert [dz | d\zeta]\vert_{\NS_j} \Big]   \right\rangle
\end{equation}
of the line bundle $\Ber (H^0(X, \omega|_{\NS}))$ and compute its image under the isomorphism \eqref{NS-Mumford1}. Here we have uniformized notation and assumed $z_j = z_j^0$ for $j \le 3$ and $\zeta_j = \zeta_j^0$ for $j \le 2$.
The vertex operators $V_j$ represent the incoming and outgoing particle states.
These are vector-valued functions, and the notation $\langle\ \rangle$ stands for the \emph{vacuum expectation value $(\!$VEV$\,)$}, i.e. $\langle O \rangle=\langle \Omega|O\Omega \rangle$, where $\Omega$ is the vacuum state, $O$ any given operator
having $\Omega$ in its domain, and $\langle | \rangle$ denotes the inner product on the Hilbert space of physical states.

\begin{thm}
\label{1}
The choice of 
conformal vertex-operator insertions \eqref{VO}
gives for any $X\in\cMn$ a section
of $\Ber H^0(X, \linebreak[0] \omega|_{\NS})$.
Then the canonical
isomorphism
\eqref{NS-Mumford1} gives a holomorphic section of the line bundle $\omega_{\cMn}$ over the dense
open subset of $\cMn$ defined by forcing each NS puncture to stay
within the open chart with bosonic coordinate $z \ne \infty$. This 
section 
can be expressed in global
coordinates $(z_4, \dots, z_n \, | \, \zeta_3, \dots \zeta_n)$ on
$\cMn$ as

\begin{equation}
  \label{main}
  \left\langle  \prod_{j=1}^n V_j(z_j | \zeta_j) \right\rangle d\nu \, ,
\end{equation}
where
\begin{equation*}
d \nu = - \tfrac{1}{2^{n-2} } \left(z_3-z_1-\tfrac 12
\zeta_3\zeta_1\right) \left(z_3-z_2 - \tfrac 12
\zeta_3\zeta_2\right) [ \, dz_4\dots dz_n\,|\,d\zeta_3\dots d\zeta_n]
\end{equation*}
and $z_j = z^0_j$, $j=1,2,3$, and $\zeta_a = \zeta^0_a$, $a=1,2$, are fixed finite values $($for any concrete purposes we can set $\zeta^0_1=\zeta^0_2=0)$.

If instead we set $z_1=\infty$, or more concretely, $(w_1 | \eta_1)
= \left(\tfrac{1}{z_1} \Big\vert \tfrac{\zeta_1}{z_1}\right) = ( 0 \, | \, \eta^0_1)$, and
replace the vertex-operator insertion $\left\langle V_1(z_1 | \zeta_1) \big[ \, \zeta [dz | d\zeta]\vert_{\NS_1}\, \mid\, [dz | d\zeta]\vert_{\NS_1} \big] \right\rangle$ with
\[
\left\langle V_1 (0|\eta_1) \big[ \, \eta [dw | d\eta]\vert_{\NS_1} \mid [dw | d\eta]\vert_{\NS_1} \big] \right\rangle \, ,
\]
then the corresponding canonical
section of $\omega_{\cMn}$ on the whole supermoduli space $\cMn\subset  \CC^{n-3 \, | \, n-2}/\ZZ_2$ becomes
\begin{equation}
  \label{Ber3}
  \left\langle V_1 (0|\eta_1)
     \prod_{j=2}^n V_j(z_j | \zeta_j) \right\rangle d\nu \, ,
\end{equation}
where
\begin{equation*}
 d \nu = \tfrac{1}{2^{n-2}}\left( 1 + \tfrac 12 \zeta_3 \eta_1 \right)
\left(z_3-z_2 -\tfrac 12 \zeta_3\zeta_2 \right) [\, dz_4\dots dz_n\,|\,d\zeta_3\dots d\zeta_n]
\end{equation*}
and $\eta_1 = \eta_1^0$, $z_j = z^0_j$, $j=2,3$, and $\zeta_2 = \zeta^0_2$.
\end{thm}

\begin{rem}
This theorem has an immediate application to physical observables.
Define the tree-level amplitude for scattering $n \ge 3$ type II superstrings in the NS sector as
\begin{equation}\label{eq:main}
    \cA_n= \int_{\cMn} \left\langle
     \prod_{j=1}^n V_j(z_j|\zeta_j;\bar z_j|\bar \zeta_j) \right\rangle 
     dm \,,
\end{equation}
where the bar indicates complex conjugation, the vertex operators combine left- and right-movers, i.e., are smooth, rather than holomorphic, sections of the corresponding bundles,  and the measure is the product of the holomorphic and antiholomorphic factors:
\[
dm=d\nu\cdot \overline{d\nu} \, ,
\]
with $d\nu$ defined in \Cref{1}. Thus we have gotten
the standard expression for the scattering amplitude, usually obtained in physics in an ad hoc way by introducing ghosts and picture changing operators. We review this process in \Cref{sec:amplitudes} to reconcile it with our formulas. However, our formulas have the advantage of also determining the closed-string coupling constant $g_c = (1/2)^2 = 1/4$ canonically, since the super Mumford isomorphism is defined canonically, not just up to scaling. 



In \eqref{eq:main} we have omitted the contribution of the dilaton field $\phi$, which usually gives $g_c=\langle e^{\phi}\rangle$. If we were to include the dilaton field, the coupling constant would have to be adjusted to  $g_c=\frac 14 \langle e^{\phi}\rangle$.
\end{rem}
\smallskip
While the holomorphic forms \eqref{main} and \eqref{Ber3} in \Cref{1} are written as an expression in terms of the chosen global coordinates on $\cMn\hookrightarrow \CC^{n-3|n-2}/\ZZ_2$, of course, the actual
forms \eqref{main} or \eqref{Ber3} are well-defined on the moduli space, independent of which coordinate system is chosen.
The advantage of our approach is that we are able to determine the holomorphic forms \eqref{main} and \eqref{Ber3} and measure $d\nu$ canonically from first principles, without any need for PCOs or ghosts --- including the determination of the coupling constant.

Beyond being a direct mathematical computation not involving ghosts or PCOs, we hope that the current paper may be of reference value in assembling together the many mathematical pieces of the construction. We also view our computations as a stepping stone towards performing similar
explicit computations for tree-level superstring scattering amplitudes with Ramond markings, where the procedure is quite different from the NS case \cite{Freeman-West}. For the Ramond case, Ott and the third author \cite{Ott-Voronov} have recently obtained an explicit geometric description of $\cM_{0,0,k}$, which should enable a similarly explicit determination of superstring scattering amplitudes.

\smallskip
{\bf Acknowledgments:} The second author is grateful to Luis
\'Alvarez-Gaum\'e for sharing his knowledge of the subject and for
literature references. The third author thanks Mikhail Kapranov and Katherine Maxwell for useful discussions and patience during the work on parallel projects.
The first and third authors would like to thank the Simons Center for
Geometry and Physics for hospitality in Spring 2023 during the
``Supergeometry and supermoduli'' program, where all the authors
started their attempt to reach a complete mathematical understanding
of genus zero superstring scattering amplitudes.

\section{Basic conventions and constructions with SRSs and their moduli}
\label{sec:foundation}

In this paper, we work with \emph{super Riemann surfaces $($SRSs$)$ of genus zero}, which all are
isomorphic to the standard SRS of genus
zero based on the projective superspace $\PP^{1|1} := \Proj \CC[u,v | \xi]$ over the
ground field $\CC$ of complex numbers, see \cite{Manin.1991.ting, Witten.2019.nosrsatm}. We will write~$X$ for such an SRS, when it is
essential to think of the \emph{SRS, or superconformal, structure}, i.e. an odd, maximally non-integrable distribution
$\cD\subset\cT_X$, see below. The supermanifold $\PP^{1|1}$
may be glued out of two affine open charts $\CC^{1|1}$ with
coordinates $(z|\zeta) = (v/u \, |\,  \xi/u)$ and $(w | \eta) = (u/v \, | \, \xi/v)$ by the gluing functions
\[
w = 1/z, \qquad \eta = \zeta/z.
\]
The \emph{odd distribution} $\cD$ defining the SRS structure is generated by the super vector field
\begin{equation}
    \label{D-zeta}
D_\zeta := \frac{\del}{\del \zeta} + \zeta \frac{\del}{\del z}
\end{equation}
in the chart $(z|\zeta)$ and the super vector field
\[
\frac{\del}{\del \eta} - \eta \frac{\del}{\del w}
\]
in the chart $(w|\eta)$. Note that $\tfrac 12 [D_\zeta, D_\zeta] = D_\zeta^2 = \del/\del z$, and this
ensures the \emph{maximal non-integrability} of $\cD$.

A remarkable feature of super Riemann surfaces is \emph{dualism between points and divisors}, making the geometry of SRSs akin to that of classical Riemann surfaces, where the fact that divisors are linear combinations of points plays a prominent role. In particular, Neveu-Schwarz (NS) punctures have their divisor avatars, and we will often switch between NS punctures and NS divisors. What creates this dualism is the structure distribution $\cD$. Given a point $p$ on a SRS $X$, we can define a complex subsupermanifold of dimension $0|1$ supported on $p$, i.e., a \emph{prime divisor}, with tangent space given by $\cD\vert_p$. In a family of once-punctured SRSs, represented by a universal SRS $\cC$ over the moduli space $\cM$ with a section $p: \cM \hookrightarrow \cC$, regarded as a subspace of the universal SRS $\cC$, the corresponding NS divisor is given as the total space,
determined by its $\cO_{p}$-algebra $S(\cD^*|_{p})$ of functions,
of the odd line bundle $\cD|_{p}$ over ${p}$. The inclusion of this total space is given by the projection of the conormal bundle $\cT_{\cC/\cM}^*\bigm|_{p}$ of $p$ in $\cC$ to $\cD^*|_{p}$, which is the dual of the inclusion $\cD \hookrightarrow \cT_{\cC/\cM}$. Conversely, given a prime divisor, we can find a unique point in it so that the tangent space to the divisor at this point is given by the distribution $\cD$. If $f = 0$ is a local equation of a prime divisor and $D$ is a generator of $\cD$ near the divisor's support, then the equations $f =0$ and $Df = 0$ determine the corresponding point. In coordinates $(z \vert \zeta)$ in which the structure distribution is generated by $D_\zeta$, see \eqref{D-zeta}, the divisor corresponding to a point $(z_0 \vert \zeta_0)$ is given by the equation $z - z_0 - \zeta \zeta_0 = 0$.








We will work with families $X \to S$ of SRSs of genus zero with NS punctures over a base $S$. We will assume the base $S$ to be affine, just like the supermoduli space. We will still denote the total space of the family by $X$.
We will also denote by $\Pi$ the parity change operator on a super vector space or sheaf of such.

The \emph{Berezinian $\Ber F$ of a free module} $F$ of finite rank over a supercommutative algebra is a rank-$(1|0)$ (or $(0|1)$, if $n$ is odd) free module, defined as the set $[e_1, \dots , \e_m \, | \, \eps_1, \dots, \eps_n]$ of equivalence classes of homogeneous bases $\{e_1, \dots , e_m \, | \, \eps_1, \dots, \eps_n\}$ of $F$ modulo the relation 
\[
[M e_1, \dots , M e_m \, | \, M\eps_1, \dots, M\eps_n] \sim (\Ber M) \cdot [e_1, \dots , \e_m \, | \, \eps_1, \dots, \eps_n]
\]
for each invertible linear automorphism $M$ of the module $F$. We also add the zero element to $\Ber F$ to make it a module. Here $\Ber M$ is the \emph{Berezinian, or superdeterminant}, of $M$. If, in a homogeneous basis, $M$ is represented by a block matrix
\[
M = \begin{pmatrix}
    A & B\\
    C & D
\end{pmatrix}
\]
with respect to the decomposition $M = M_0 \oplus M_1$ into the even and odd parts,
then
\[
\Ber M := \det(A-BD^{-1}C)\det(D)^{-1} = \det(A)\det(D-CA^{-1}B)^{-1}.
\]
The Berezinian of a map of free modules is multiplicative: $\Ber (M_1 M_2) = \Ber (M_1) \Ber (M_2)$. Given a short exact sequence of free modules
\[
0 \to F_1 \to F_2 \to F_3 \to 0,
\]
we also get a canonical isomorphism of the Berezinians of the modules:
\[
\Ber F_2 \xrightarrow{\sim} \Ber F_1 \otimes \Ber F_3.
\]
A standard argument extends the construction of $\Ber F$ to locally free modules and sheaves.

\emph{Line bundles} in the super setting will be divided into \emph{even} and \emph{odd} ones, depending on whether their rank is $1|0$ or $0|1$.

We will adopt the notation $\cL^k$ for the $k$'th tensor power $\cL^{\otimes k}$ of a line bundle $\cL$ for $k \in \ZZ$. We will routinely use abbreviated notation for cohomology as well. For example, $H^i(\omega^k)$ will denote the cohomology space $H^i(\cCn, \omega_{\cCn/\cMn}^{\otimes k})$, where $\cCn \xrightarrow{p} \cMn$ is the universal family of genus-$0$ SRSs with~$n$ NS punctures. In full generality, we should be working with the space $H^0(\cMn, R^i p_* \omega_{\cCn/\cMn}^{\otimes k})$, but it is indeed naturally  isomorphic to $H^i(\cCn, \omega_{\cCn/\cMn}^{\otimes k})$, given that $\cMn$ is affine, being the $\ZZ_2$-quotient of the complement of hyperplanes in $\CC^{0|1} \times (\CC^{1|1})^{n-4}$. Likewise, the Berezinian of cohomology $B(\cF) = \Ber H^*(\cF)$ will be identified with $\Ber H^0(\cMn, \Ber ( R^\bullet p_* \cF))$. See the details on the determinant and Berezinian of cohomology in \cite{knudsen-mumford:det, deligne:detcoh} and \cite{VorSuperMum, Diroff}, respectively.

We will denote by $(z_1|\zeta_1),\dots,(z_n|\zeta_n)$ the coordinates
of the NS punctures.
To simplify notation, we will write $\omega$ for the \emph{relative
dualizing sheaf} $\omega_{\cCn/\cMn}$, and, when relevant, also for its restriction, also known as the canonical bundle, to
an individual SRS.

\section{The
super Mumford form
}\label{sec:bases}
We recall that the super Mumford isomorphism is a canonical isomorphism $B(\omega^3)= B(\omega)^{5}$ over the supermoduli space $\cM_{g,0,0}$ for $g \ge 2$, \cite{deligne:lost,VorSuperMum}. Here we will use the same construction to establish a canonical isomorphism over $\cMn$: $B(\omega^3)= B(\omega)^{5}$. To start out, and also for further use, we construct natural bases for the spaces of global sections of the powers of the relative dualizing sheaf $\omega$.
  
\subsection{Bases of cohomology}
We denote by $t := z \vol^{-1} \in H^0(\cCn, \omega^{-1})$ the generating section of $\omega^{-1}$, and consider, for any positive integer~$k$, the following short exact sequence on $\cCn$:
\[
0 \to \omega^k  \xrightarrow{\cdot t} \omega^{k-1} \to \omega^{k-1}|_{\div t} \to 0.
\]
We make the maps in the associated long exact sequence of homology explicit for $k=1,2,3$, by first writing out the bases of the relevant spaces, and then identifying the maps on homology explicitly in these bases, so that $B_k(t)$ is the canonical isomorphism of the associated Berezinians $B(\omega^{k-1})\simeq \Pi B(\omega^k)\otimes B( \omega^{k-1}|_{\div t} )$. These, suitably modified in the presence of NS markings, will be used for the computation of the super Mumford form.

{\bf Case $k=1$.} Here we get the exact sequence (of cohomology of sheaves over $\cCn$) 
\[ 0=H^0(\omega)\to H^0(\cO)\to H^0(\cO|_{\div t})\to H^1(\omega)\to H^1(\cO)=0\]
The basis of $H^0(\cO)$ is $\{ 1 \, | \; \}$ --- this notation is to list the even basis vectors followed by the odd basis vectors (none in this case) after the dash. Since $H^i(\cO)=0$ for $i>0$, it follows that the Berezinian $B(\cO)=\Ber H^0(\cCn;\cO) = \det H^0(\cCn;\cO)$ has as basis the section $\Ber \{1\, |\;\} = \det\{1\, |\;\}$, which we will denote by $[1\, |\;]$. Using relative super Serre duality, we see that $H^1(\omega)=H^0(\cO)^*$ on each genus zero SRS $X$, which globalizes since we have global coordinates on the affine superscheme $\cMn$.
Explicitly, the Serre dual to the even basis element $1$ of $H^0(\cO)$ is given by $\left\{ \left. \tfrac{\zeta}{z} \vol \; \right| \ \right\}$, and we observe that $B(\omega)=(\Ber(H^1(\omega)))^{-1}=\Ber(H^0(\cO))=B(\cO)$ also has basis $[1 \, | \; ]$.

The basis of $H^0(\cO|_{\div t})$ is $\{ 1|_{\div t} \, | \; \zeta|_{\div t}\}$, and thus $[1|_{\div t} \, | \; \zeta|_{\div t}] \in B(\cO|_{\div t})$ is the generating section of that sheaf. We can thus compute
\[ B_1(t) := [1 \, | \; ] \otimes \Pi [1 \, | \; ] \otimes [1|_{\div t} \, | \; \zeta|_{\div t}]^{-1} \in B(\cO) \otimes \Pi B(\omega) \otimes B(\cO|_{\div t})^{-1}\,.\]

{\bf Case $k=2$.} Here the exact sequence of cohomology is
\[ 0=H^0(\omega^2)\to H^0(\omega)\to H^0(\omega|_{\div t})\to H^1(\omega^2)\to H^1(\omega)\to H^1(\omega|_{\div t})\to0\]
(where the last zero is simply for dimensional reasons). An explicit basis of $H^1(\omega^2) = H^0(\omega^{-1})^*$ can be written simply by thinking of these as suitable volume forms with poles:
\[
\left\{ \frac{\vol^2}{z} \; \left| \; \frac{\zeta \vol^2}{z}, \frac{\zeta \vol^2}{z^2} \right\} = \left\{ \frac{\zeta}{\vol} \; \right| \; \frac{1}{\vol}, \frac{z}{\vol} \right\}^*\,.
\]
Thus the generating section of the Berezinian $B(\omega^2)=B(\omega^{-1})$ is given by 
\[
\left[\frac{\zeta}{\vol} \; \left| \; \frac{1}{\vol}, \frac{z}{\vol}\right. \right]
\]
The basis of $H^0(\omega|_{\div t})$ is given by $\{ \zeta \vol|_{\div t} \; | \; \vol|_{\div t}\}$, 
and we finally compute
\begin{multline*}
B_2(t) := \Pi [1 \, | \; ]^{-1} \otimes \left[\frac{\zeta}{\vol} \; \left| \; \frac{1}{\vol}, \frac{z}{\vol}\right. \right]^{-1} \otimes \left[\zeta \vol|_{\div t} \; \left| \; \vol|_{\div t} \right. \right]\\
\in \Pi B(\omega)^{-1} \otimes B(\omega^2)^{-1} \otimes B(\omega|_{\div t}).
\end{multline*}

{\bf Case $k=3$.} In this case a basis of $H^1(\omega^3) = H^0(\omega^{-2})^*$ can be given as 
\[
\left\{ \zeta,\zeta z,\zeta z^2| z,z^2\right\} \cdot  \frac{\vol^3}{z^3}\,,
\]
(where from now on we write such bases by separating out the common factor, for readability)
and we have the generating section
\[
\left[\frac{\zeta\vol^3}{z}, \frac{\zeta\vol^3}{z^2}, \frac{\zeta\vol^3}{z^3}  \; \left| \; \frac{\vol^3}{z}, \frac{\vol^3}{z^2} \right. \right] \in B(\omega^3) = B(\omega^{-2}).
\]
Similarly by thinking of volume forms with poles, we see that $\{ \vol^2|_{\div t} \; | \; \zeta\vol^2|_{\div t}\}$ gives a basis of $H^0(\omega^2|_{\div t})$, and thus we finally compute
\begin{multline*}
B_3(t) := \left[\frac{\zeta}{\vol} \; \left| \; \frac{1}{\vol}, \frac{z}{\vol}\right. \right]\\
 \otimes \Pi \left[\frac{\zeta\vol^3}{z}, \frac{\zeta\vol^3}{z^2}, \frac{\zeta\vol^3}{z^3}  \; \left| \; \frac{\vol^3}{z}, \frac{\vol^3}{z^2} \right. \right]\\
 \otimes \left[ \vol^2|_{\div t} \; \left| \; \zeta\vol^2|_{\div t} \right.\right]^{-1}
\in B(\omega^2) \otimes \Pi B(\omega^3) \otimes B(\omega^2|_{\div t})^{-1}.
\end{multline*}

\subsection{A formula for the super Mumford form}
We are now ready to determine explicitly the super Mumford isomorphism, recalling that unlike the bosonic case, this is not just an abstract isomorphism of 
line bundles, but in fact the canonical isomorphism given by the element $\mu:=B_3(t)\otimes B_2(t)\otimes B_1(t)^{-2}\in B(\omega^3)\otimes B(\cO)^{-5}$, where to identify this product we recall that $B(\cO)$ is canonically isomorphic to $B(\omega)$. To obtain an expression for $\mu$, we manipulate the formulas for $B_1(t),B_2(t),B_3(t)$ from the previous section. We obtain
\[
B_2(t)\otimes B_1(t)^{-1} = \frac{ \left[\zeta \vol|_{\div t} \; \left| \; \vol|_{\div t} \right. \right] \otimes [1|_{\div t} \; | \; \zeta|_{\div t}]}{ [1 \, | \; ]^3  \otimes \left[\frac{\zeta}{\vol} \; \left| \; \frac{1}{\vol}, \frac{z}{\vol}\right. \right]}
\]
Recall that restricting any even line bundles $\cL$ and $\cM$ on a SRS~$X$ to the divisor $\div t$, we have canonical isomorphisms
\[
B(\Pi \cL|_{\div t}) = B(\cL|_{\div t})^{-1};\quad B(\cL|_{\div t}) = B(\cM|_{\div t})\,,
\]
where the latter identification is given by multiplying by any generating holomorphic section of $\cL^{-1} \otimes \cM$ in a neighborhood of $\div t$. Applying this to all the restrictions in the formula above, we compute
\[
B_2(t)\otimes B_1(t)^{-1} = \frac{1}{[1 \, | \; ]^3  \otimes \left[\frac{\zeta}{\vol} \; \left| \; \frac{1}{\vol}, \frac{z}{\vol}\right. \right]}  \in B(\cO)^{-3} \otimes B(\omega^2)^{-1}.
\]
By a similar simplification of the expression for $B_3(t)$ from the previous section, we obtain
\begin{multline*}
B_3(t)\otimes B_1(t)^{-1}
= \frac{\left[\frac{\zeta}{\vol} \; \left| \; \frac{1}{\vol}, \frac{z}{\vol}\right. \right] \otimes \left[\frac{\zeta\vol^3}{z}, \frac{\zeta\vol^3}{z^2}, \frac{\zeta\vol^3}{z^3}  \; \left| \; \frac{\vol^3}{z}, \frac{\vol^3}{z^2} \right. \right] }{[1 \, | \; ]^2 }\\
\in B(\omega^2) \otimes B(\omega^3)\otimes B(\cO)^{-2}.
\end{multline*}
Finally, noticing that the canonical section $[1\;|\;]$ defines a global trivialization of $B(\cO)$, and can thus be omitted in the formulas above, we obtain the following expression for the super Mumford form
\begin{equation}
\label{superMum}
\mu = B_3(t)\otimes B_2(t)\otimes B_1(t)^{-2} = \left[\tfrac{\zeta\vol^3}{z}, \tfrac{\zeta\vol^3}{z^2}, \tfrac{\zeta\vol^3}{z^3}  \; \left| \; \tfrac{\vol^3}{z}, \tfrac{\vol^3}{z^2} \right. \right]
\end{equation}
as a global section of the line bundle $B(\omega^3)\otimes B(\cO)^{-5} = B(\omega^3)$ over $\cMn$.

\section{The \sKS map over $\cMn$
}\label{sec:sKS}
We now proceed to work out the super Mumford isomorphism for the case of genus 0 SRS with NS markings, by computing the natural isomorphism given by the \sKS map. 

The superanalog of the Kodaira-Spencer map in the case of deformation theory of supermanifolds is well-known \cite{Manin.1986.criticaldim, LeBrun.Rothstein.1988.moduli, FKP2}. In our case this will be a map from $\cT_X \cMn$, which is a super vector space of dimension $n-3|n-2$, to a suitable first homology. For our purposes it will be more natural to work with the dual map, as it is more explicit to write bases for spaces of sections~$H^0$ rather than for the first homology~$H^1$. We first recall the setup, and then compute the \sKS map explicitly, by using an explicit cover to get the class in \v{C}ech cohomology. 

\subsection{The construction}
For a SRS $X$, let $\cT^{sc}_X \subset \cT_X$ be the sheaf of
\emph{superconformal vector fields}, i.e.~vector fields that preserve
the distribution $\cD$, i.e.~of $v\in\cT_X$ such that
$[v,\cD]\subset\cD$. 
coordinates, see formula (12) in Katherine Maxwell's paper
https://arxiv.org/pdf/2002.06625.pdf. Note that $\cT^{sc}_X$ is {\em
  not} an $\cO_X$-module: the product $f\cdot v$ of $v\in\cT_X^{sc}$
and a (local holomorphic) function $f\in\cO_X$ may not lie in
$\cT^{sc}_X$. However, as a sheaf of abelian groups, $\cT^{sc}_X$ is
isomorphic to $\cD^{\otimes 2}$, and in fact this natural isomorphism
can be used to induce an $\cO_X$-module structure on $\cT^{sc}_X$, see
\cite{Maxwell}. We thus consider $\cT^{sc}_X(-\NS)\subset \cT^{sc}_X$
as the subsheaf of those vector fields that preserve $\cD$, and which
also do not move the
NS divisor $\NS$.

The previous discussion can be globalized to consider $\cD$ as a subsheaf of the relative tangent sheaf $\cT_{\cCn / \cMn} \subset \cT_{\cCn}$ and then 
define the subsheaves $\cT^{sc}_{\cCn/\cMn}(-\NS) \subset \cT_{\cCn / \cMn}$ and $\cT^{sc}_{\cCn}(-\NS) \linebreak[0] \subset \cT_{\cCn}$
of vector fields that preserve the global $\cD$ and fix the global divisor $\NS$, which we think of as an unramified effective Cartier divisor.

Altogether, this yields the exact sequence of sheaves of vector spaces
$$
 0\to\cT^{sc}_{\cCn/\cMn}(-\NS)\to \cT^{sc}_{\cCn}(-\NS)\to p^*\cT_{\cMn}\to 0\,.
$$
Then the \emph{\sKS isomorphism} is the connecting homomorphism
\begin{equation*}
\cT_{\cMn} \to R^1 p_* \cT^{sc}_{\cCn/\cMn}(-\NS) = R^1 p_*\cD^{\otimes 2}(-\NS) 
\end{equation*}
induced by the isomorphism 
\begin{equation}
\label{sc}
\begin{split}
\cD^{\otimes 2} & \xrightarrow{\sim} \cT^{sc}_{\cCn/\cMn},\\
f D_\zeta \otimes D_\zeta & \longmapsto f\partial_z +\tfrac {(-1)^{\tilde f}}2 D_\zeta (f) D_\zeta,
\end{split}
\end{equation}
see \cite{Beilinson.Manin.Schechtman.1987, Manin.1988,RSV.1988,Maxwell}.
Taking global sections, we globalize the \sKS isomorphism to the isomorphism
\[
\KS: H^0 (\cMn, \cT_{\cMn} )\rightarrow H^1 (\cCn, \cD^{\otimes 2}(-\NS) ).
\]
The \emph{\sKS isomorphism over a given SRS} $X\in\cMn$ is
$$
\cT_{\cMn} \vert_X \to H^1( X, \cD \vert_X^{\otimes 2}(-\NS) )
.
$$
To identify $H^1 (\cCn, \cD^{\otimes 2}(-\NS) )$ further, we apply \emph{Serre duality}:
\begin{equation*}
R^1 p_*\cD^{\otimes 2}(-\NS) = p_* \omega_{\cCn/\cMn}^{\otimes 3}(\NS)^*,
\end{equation*}
which leads to the isomorphism
\[
\sigma \co H^1 (\cCn, \cD^{\otimes 2}(-\NS) ) = H^0 (\cCn, \omega^3 (\NS))^*.
\]
Combined with the \sKS isomorphism, this gives the dual isomorphism:
\[
p_* \omega_{\cCn/\cMn}^{\otimes 3}(\NS) \to \cT_{\cMn}^*,
\]
which globalizes to the \emph{dual \sKS isomorphism}
\[
\KS^* \co H^0 (\cCn, \omega_{\cCn/\cMn}^{\otimes 3}(\NS)) \to H^0 (\cMn, \cT_{\cMn}^*).
\]

\subsection{Proof of the main theorem}

The rest of the paper is largely dedicated to the explicit computation of the \sKS map and identification of bases for the spaces involved. Before we get into that, we would like to outline the proof of \Cref{1}. Consider the following short exact sequence (SES) of sheaves over $\cCn$, by restricting to the NS divisor:
\begin{equation}
\label{eq:preM}
 0\to\omega^3{\rightarrow}\omega^3(\NS){\rightarrow}\omega^3(\NS)|_{\NS}\to 0\,.
\end{equation}
This is the analog for a divisor with multiple components of the case $k=3$ of the SES described in \Cref{sec:bases}. This SES gives a canonical isomorphism of the Berezinian line bundles:
\[
B(\omega^3) \otimes B(\omega^3(\NS)|_{\NS}) \xrightarrow{\sim} B(\omega^3(\NS))\,,
\]
which we are going to identify explicitly. Since $H^0(\omega^3)=0$, and $H^1(\omega^3(\NS)|_{\NS})=0$, as the relative (topological) dimension of the $\NS$ divisor over the moduli space $\cMn$ is 0 and $\cMn$ is affine, the associated exact sequence in cohomology is simply
\begin{equation}
\label{eq:M}
    0\rightarrow H^0(\omega^3(\NS)) 
    {\rightarrow} H^0(\omega^3(\NS)\vert_{\NS}) 
    {\rightarrow} H^1(\omega^3)\rightarrow 0\,.
\end{equation}
The choice of any even generating section of the line bundle $\omega^2(\NS)$ in an affine open neighborhood of the divisor $\NS$ yields an isomorphism
\begin{equation}
\label{core}
B(\omega|_{\NS}) \xrightarrow{\sim} B(\omega^3(\NS)|_{\NS}) \, ,
\end{equation}
which does not depend on the choice of the section and is thereby canonical, see \cite{VorSuperMum,Diroff}.
provided none of the NS punctures is at $\infty$, and
the Berezinian line bundle
$B(\omega^3(\NS)\vert_{\NS})$. Thus, we have a canonical isomorphism
\[
\varphi: B(\omega^3) \otimes B(\omega|_{\NS}) \xrightarrow{\sim} B(\omega^3(\NS))\,.
\]

In \Cref{app:mapM} we provide further details and compute
the image of the super Mumford form $\mu \in \Gamma(\cMn, B(\omega^3))$, given by
$\eqref{superMum}$, tensored with a global section of $B(\omega|_{\NS})$ arising from vertex-operator insertions, under the isomorphism $\varphi$, see \Cref{prop}:
\begin{multline*}
\varphi\left(\mu \otimes \left\langle
     \bigotimes_{j=1}^n V_j(z_j | \zeta_j) \Big[ \, \zeta [dz | d\zeta]\vert_{\NS_j} \Big| [dz | d\zeta]\vert_{\NS_j} \Big]   \right\rangle \right)\\
     = \left\langle
     \prod_{j=1}^n V_j(z_j | \zeta_j)   \right\rangle \Ber \left(\{\zeta,\zeta z,\dots, \zeta
z^{n-4}\,\vert\, 1,\dots, z^{n-3}\} \cdot \frac
{[dz|d\zeta]^3}{\prod_{j=1}^n (z-z_j-\zeta \zeta_j)} \right)\,,
\end{multline*}
where $(z_1 | \zeta_1)$, \dots, $(z_1 | \zeta_1)$ are the NS
punctures. In \Cref{VOI} we explain how vertex-operator insertions make up a section of $B(\omega|_{\NS})$. The Berezinian of the dual \sKS isomorphism, computed in
the remaining part of \Cref{sec:sKS}, with detailed linear algebra
computations given in \Cref{app:BerMn}, naturally identifies
$B(\omega^3(\NS))$ with the space $H^0(\cMn, \omega_{\cMn}) =
H^0(\cMn, \Ber (\cT_{\cMn}^*))$ of holomorphic volume forms on $\cMn$:
\[
\Ber (\KS^*) \co B(\omega^3(\NS)) \xrightarrow{\sim} H^0(\cMn,  \omega_{\cMn}).
\]

\Cref{prop2} below gives an explicit computation,
\[
d \nu := - \frac{1}{2^{n-2}} (z_3-z_1 - \tfrac 12 \zeta_3\zeta_1 ) (z_3-z_2 - \tfrac 12 \zeta_3\zeta_2 ) [dz_4\dots dz_n\,|\,d\zeta_3\dots d\zeta_n] \, ,
\]
of the image
\begin{equation*}
\Ber (\KS^*) \left(\Ber \left(\{\zeta,\zeta z,\dots,
\zeta z^{n-4}\,\vert\, 1,\dots, z^{n-3}\} \cdot \frac
      {[dz|d\zeta]^3}{\prod_{j=1}^n (z-z_j-\zeta \zeta_j)} \right)\right)
\end{equation*}
under the isomorphism $\Ber (\KS^*)$, when all the NS punctures are finite. and yields \Cref{1}.
The case $z_1 = \infty$ is
treated in \Cref{infty}.
   \qed


\subsection{The \v{C}ech cohomology
computation
}
First, note that for computations with supermanifolds, it is enough to take an acyclic (Leray) cover of the underlying bosonic manifold. Below is a precise statement, which is surely known to seasoned supergeometeres, though we could not find it in the literature.

For a topological space $X$ and a sheaf $\cF$ on it, recall that an open cover $\{U_i\}$ is called \emph{Leray with respect to $\cF$} if for the derived-functor cohomology, we have $H^q (U_{i_1} \cap \dots \cap U_{i_p}, \cF) = 0$ for all finite sets $\{i_1, \dots , i_p\}$ of indices and all $q > 0$. In this case, the derived-functor cohomology $H^q(X, \cF)$ will agree with the \v{C}ech coholomogy of $\cF$ with respect to this cover, justifying using the same notation for both.

\begin{lem}[Folklore]
\begin{enumerate}[$(1)$]
    \item If $X$ is a separated superscheme of finite type and $\cF$ is a quasi-coherent sheaf on $X$, then any cover of the underlying scheme $X_{\red}$ by affine open schemes is Leray with respect to $\cF$.
    \item If $X$ is a complex supermanifold  and $\cF$ is a coherent sheaf on $X$, then any cover of the underlying manifold $X_{\red}$ by open Stein manifolds is Leray with respect to $\cF$.
\end{enumerate} 
\end{lem}

This lemma follows from the following lemma.

\begin{lem}[Folklore]
\begin{enumerate}[$(1)$]
    \item If $X$ is a superscheme of finite type whose reduction $X_{\red}$ is affine and $\cF$ is a quasi-coherent sheaf on $X$, then $H^q(X, \cF) = 0$ for all $q > 0$.
    \item If $X$ is a complex supermanifold whose reduction $X_{\red}$ is Stein and $\cF$ is a coherent sheaf on $X$, then $H^q(X, \cF) = 0$ for all $q > 0$.
\end{enumerate} 
\end{lem}

\begin{proof}
(1) Let $J_X = ((\cO_X)_1)$ be the ideal sheaf of odd nilpotents on the superscheme $X$. If $X$ is of finite type and $\cF$ is a quasi-coherent sheaf on it, then $\cF$ is a finite iterated extension by quasi-coherent sheaves $J_X^n \cF/J_X^{n+1} \cF$ of $\cO_{X_{\red}} = \cO_X/J_X$-modules, whose higher cohomology vanishes by the algebraic version of Cartan's theorem B, see \cite[Theorem III.3.7]{Hartshorne}. Therefore, by downward induction on $n$, the higher cohomology of $\cF$ on $X$ will also vanish.

(2) The same argument works, except that Cartan's theorem B itself is applied to the coherent sheaves $J_X^n \cF /J_X^{n+1} \cF$ on $X_{\red}$.
\end{proof}

\begin{rem}
By a straightforward super generalization of Serre's theorem on affineness, it follows that if $X$ is a superscheme of finite type whose reduction $X_{\red}$ is affine, then $X$ is affine as a superscheme.
\end{rem}

To make formulas more symmetric, we will perform the computation below
while setting the points $z_1^0,z_2^0,z_3^0$ to be finite. (The case
$z_1^0 = \infty$ is treated in \Cref{infty}.) Recall that while
all these 3 bosonic coordinates along with the corresponding fermionic
coordinates $\zeta_1^0$ and $\zeta_2^0$ can be fixed by the action of
$\operatorname{OSp}(1|2)^0$, the $\zeta_3$ remains free, and in
anticipation of this we will be taking one ``unnecessary'' open set,
$W_3$, to handle this.

We will 
compute the \sKS map
\[
\KS: H^0 (\cMn, \cT_{\cMn} )\rightarrow H^1 (\cCn, \cD^{\otimes 2}(-\NS))
\]
in \v{C}ech cohomology. Since the \sKS map is a morphism of sheaves over $(\cMn)_{\red} =  M_{0,n} $, we may compute it locally on $M_{0,n}$, which
may be identified with an open subset of $\CC^{n-3}$. For any point $(z_4^0,\dots,z_n^0)\in M_{0,n}$, we can choose $\varepsilon>0$ sufficiently small so that the ball $B_{\tfrac \varepsilon2}(z_4^0,\dots, z_n^0)$ in $\CC^{n-3}$ of radius~$\varepsilon/2$ around this point is contained in $M_{0,n}$, which is to say that within this ball all the $n-4$ bosonic points stay distinct, and also distinct from $z_1=z^0_1$, $z_2=z^0_2$ and $z_3=z^0_3$.
Then, working over this ball in the universal family $p \co \cCn \to \cMn$ of SRSs, we cover the preimage $p_{\red}^{-1}(B_{\tfrac \varepsilon2}(z_4^0,\ldots,z_n^0))\subset (\cCn)_{\red}$
by the following open subsets.

First recall that the fibers of $p_{\red}$ are simply copies of $\PP^{1}$, since in the universal curve the point is allowed to coincide with one of $z_1^0,\dots,z_n^0$. We denote by $U_j\subset\CC$, $j=3,\ldots,n$, the open radius-$\varepsilon$ disk around $z_j^0$ (adjusting $\varepsilon$, if necessary,
to make sure that the disks are all disjoint), and denote by $U_0$ the complement in $\PP^1$ of the union of
closed radius-$\varepsilon/2$
disks around these points. 
Then $\PP^1=U_0 \cup U_3\cup\dots \cup U_n$, and we now construct a suitable cover of $p_{\red}^{-1}(B_{\tfrac \varepsilon 2}(z_4^0,\ldots,z_n^0))$
by the open sets $W_\bullet:=U_\bullet\times B_{\tfrac \varepsilon2}(z_4^0,\ldots,z_n^0)$. We see that the open cover $\{ W_0, W_3, \dots , W_n\}$ is
such that the only nonempty intersections are the pairwise intersections between $W_0$ and $W_j$, since $\varepsilon$ was sufficiently small ($\varepsilon<{\rm min}\{|z_j-z_k|,\ j<k\}$). 
Moreover, these pairwise intersections are simply open annuli, and thus Stein, and thereby acyclic for coherent sheaves. We further take the open subsupermanifolds of $\cCn$ based on each of these open subsets.

We utilize the same coordinates $(z_4,\ldots,z_n|\zeta_3,\ldots,\zeta_n)$ on the superdomain in $\cMn$ over the common factor $B_{\frac \epsilon2}(z_4^0,\ldots,z_n^0)$ of all $W_j$'s and augment them with coordinates built on the standard affine coordinates $(z|\zeta)$
on the superdomains in
$\PP^{1|1}$ over the factors $U_j$, $3 \le j \le n$.
More precisely, we put 
over $W_j$ the coordinates
\begin{align*}
(z_{4},\ldots,z_{n}|\zeta_{3},\ldots,\zeta_{n};w_j|\eta_j)\,,
\end{align*}
where
\begin{equation}
\label{change}
\begin{split}
    (w_3|\eta_3)=&(z-z_3^0 - \zeta \zeta_3| \zeta-\zeta_3), \\
    (w_j|\eta_j)=&(z-z_j-\zeta \zeta_j| \zeta-\zeta_j), \quad j\geq 4.
\end{split}
\end{equation}

Since our open subsets and their intersections (pairwise, as triple are all empty) are Stein manifolds,
they are acyclic for cohomology with coefficients in coherent sheaves of $\cO$-modules on the corresponding open subsupermanifolds. Therefore this atlas can be used to compute \v{C}ech cohomology. We thus compute
the \sKS map from \eqref{change} in this atlas as
\begin{equation*}
\begin{split}
 \del_{z_j} \stackrel{\KS}{\longmapsto}
 &\Big\{ -\delta_{kj} \partial_{w_j} \in \Gamma(W_0 \cap W_k,\cT^{sc}_{\cCn/\cMn}(-\NS)) \Big\}_{k=3}^n, \qquad j=4,\ldots,n, \\
\del_{\zeta_a} \stackrel{\KS}{\longmapsto}    
&\Big\{ - \delta_{ka} (\partial_{\eta_a}-\zeta \partial_{w_a}) = \delta_{ka} ((2\eta_a+\zeta_a)\partial_{w_a}-D_{\eta_a}) \in \Gamma(W_0 \cap W_k,\cT^{sc}_{\cCn/\cMn}(-\NS)) \Big\}_{k=3}^n,
\\
& \qquad  a=3,\dots, n.
    \end{split}
\end{equation*}
Composing with the isomorphism \eqref{sc}, we get
\begin{equation}
\label{KS}
\begin{split}
 \del_{z_j} \stackrel{\KS}{\longmapsto} 
 &\Big\{ -\delta_{kj} D_{\eta_j} \otimes D_{\eta_j} \in \Gamma(W_0 \cap W_k,\cD^{\otimes 2} (-\NS)) \Big\}_{k=3}^n, \qquad j=4,\ldots,n, \\
\del_{\zeta_a} \stackrel{\KS}{\longmapsto}    
&\Big\{ \delta_{ka} ((2\eta_a+\zeta_a) D_{\eta_a} \otimes D_{\eta_a} ) \in \Gamma(W_0 \cap W_k,\cD^{\otimes 2} (-\NS)) \Big\}_{k=3}^n, \qquad  a=3,\dots, n.
    \end{split}
\end{equation}

\subsection{The dual map
}
While we have an explicit basis for $H^0(\omega^3(\NS))$, we want to compute the \sKS map in coordinates $(z_4,\ldots,z_n|\zeta_3,\ldots,\zeta_n)$ on $\cMn$.
Equivalently, we can make explicit the dual map isomorphism
\[
\KS^* \co H^0(\omega^3(\NS)) \to H^0(\cMn, \cT_{\cMn}^*)\]
obtained by combining the \sKS map with Serre duality and applying linear duality. 
The Serre duality isomorphism
$$
\sigma: H^1(X,\cD^{\otimes 2}(-\NS))
\longrightarrow H^0(\omega^{\otimes 3}(\NS))^*
$$
is given
by
\begin{align*}
    \sigma(f D_\zeta \otimes D_\zeta)(\lambda)=\Res\left[\iota_{fD_\zeta \otimes D_\zeta} \lambda\right].
\end{align*}
The contraction can be easily computed once we notice that
\begin{equation}\label{eq:contract}
    \iota_{D_\zeta}[dz|d\zeta]= - 1.
\end{equation}
A nice way to prove \eqref{eq:contract} is by using the exact sequence 
\[
0 \to \mc{D} \to \cT \to \cT/\mc{D} \to 0
\]
and the isomorphism $\mc{D}^2 = \mc{T}/\mc{D}$.  Indeed, note that
$\left[\tfrac{\del}{\del z} \, | \, \tfrac{\del}{\del\zeta}\right]$ is
a local basis of $\Ber \mc{T}$, and compute
\[
\left[\tfrac{\del}{\del z} \, | \, \tfrac{\del}{\del \zeta}\right] = \left[\tfrac{\del}{\del z} \, | \, D_\zeta\right] = [D_\zeta^2 \, | \, D_\zeta] = D_\zeta \in \Ber \mc{T} = \mc{D}.
\]
Therefore, $[dz|d\zeta] = D_\zeta^* \in \mc{D}^* = \Ber \mc{T}^* = \omega_X$, which means $[dz|d\zeta] ( D_\zeta) = 1$ and implies \eqref{eq:contract}.

So, the composition of Serre duality with the \sKS map \eqref{KS} gives
\begin{equation}
\label{KSdual}
\begin{split}
 \sigma \circ \KS \co   \partial_{z_j} &\longmapsto \sigma(- D_\zeta \otimes D_\zeta |_{W_0 \cap W_j}), \quad j=4,\ldots, n,\\
 \sigma \circ \KS \co   \partial_{\zeta_a} &\longmapsto \sigma((2\eta_a+\zeta_a) D_\zeta \otimes D_\zeta |_{W_0 \cap W_a}), \quad a=3,\ldots, n,
\end{split}
\end{equation}
where we take into account that $D_{\eta_j} = D_\zeta|_{W_0 \cap W_j}$, which follows from \eqref{change}.
We want to evaluate $\sigma$ explicitly on the basis \eqref{basisA} for $H^0(\omega^3(\NS))$, described in \Cref{prop}. For convenience, we rewrite that basis in the form
\begin{equation}
\label{basis}
 \left\{\zeta z^a  \Bigg\vert\, z^b \left(1-\sum_{k=1}^n \frac {\zeta_k \zeta}{z-z_k} \right) \right\} \cdot \frac{\vol^3}{\prod_{j=1}^n (z-z_j)},
\end{equation}
where $b=0,\ldots,n-3$ and $a=0,\ldots,n-4$. In this formula, and
until the end of this section, to simplify notation, we will simply
write $z_1,z_2,z_3$ instead of $z_1^0,z_2^0,z_3^0$ and
$\zeta_1,\zeta_2$ instead of $\zeta_1^0,\zeta_2^0$. We stress that
these are fixed points in $\CC$, not free bosonic coordinates on
$\cMn$ like $z_4,\dots,z_n$ (and we still do not allow them to
take infinite value, except in \Cref{infty}). However, as they appear
symmetrically in the formulas, we will avoid using the superscript $0$
here and in the Appendices.

We compute
\begin{equation*}
    \sigma\left(- D_\zeta \otimes D_\zeta |_{W_0\cap W_i} \right)\left(\frac{z^a \zeta [dz|d\zeta]^3}{\prod_{j=1}^n (z-z_j)}\right) =
    \Res|_{z=z_i}\left[ \frac{-z^a \zeta [dz|d\zeta]}{\prod_{j=1}^n (z-z_j)} \right]=\frac {-z_i^a}{\prod_{k\neq i} (z_i-z_k)},
\end{equation*}
\begin{align*}
    \sigma\left(- D_\zeta \otimes D_\zeta |_{W_0\cap W_i} \right)\left(\frac {z^b [dz|d\zeta]^3}{\prod_{j=1}^n (z-z_j)} \left(1-\sum_{k=1}^n \frac {\zeta_k \zeta}{z-z_k} \right)\right)=& \Res|_{z=z_i}\left[\frac {z^b [dz|d\zeta]}{\prod_{j=1}^n (z-z_j)} \left( \sum_{k=1}^n \frac {\zeta_k \zeta}{z-z_k} \right) \right]\\
    = \frac {z_i^b}{\prod_{j\neq i} (z_i-z_j)} \sum_{k\neq i} \frac {\zeta_k}{z_i-z_k}& + \zeta_i \frac {d}{dz}_{z=z_i} \frac {z^b}{\prod_{j\neq i} (z-z_j)}\\
    =\frac {z_i^b}{\prod_{j\neq i} (z_i-z_j)} \sum_{k\neq i} \frac {\zeta_k-\zeta_i}{z_i-z_k}& + \frac {b \zeta_i z_i^{b-1}}{\prod_{j\neq i} (z_i-z_j)}
\end{align*}
for $i=4,\ldots,n$ and, similarly,
\begin{align*}
    \sigma\left((2\eta_c+\zeta_c) D_\zeta \otimes D_\zeta |_{W_0\cap W_c} \right)\left(\frac {z^a \zeta [dz|d\zeta]^3}{\prod_{j=1}^n (z-z_j)}\right) & = 
    \frac {- z_c^a \zeta_c}{\prod_{k\neq c} (z_c-z_k)},\\[0.4cm]
    \sigma\left((2\eta_c+\zeta_c) D_\zeta \otimes D_\zeta |_{W_0\cap W_c} \right)\left(\frac {z^b [dz|d\zeta]^3}{\prod_{j=1}^n (z-z_j)} \left(1-\sum_{k=1}^n \frac {\zeta_k \zeta}{z-z_k} \right)\right) & = \frac {2 z_c^{b}}{\prod_{j\neq c} (z_c-z_j)}\cr
    & \quad +\frac {z_c^b}{\prod_{j\neq c} (z_c-z_j)} \sum_{k\neq c} \frac {\zeta_c \zeta_k}{z_c-z_k}
\end{align*}
for $c=3,\ldots,n$.

\subsection{The Berezinian of the dual
map
}
\begin{prop}
\label{prop2}
The isomorphism $\Ber (\KS^*)$ maps $\Ber \left(\{\zeta,\zeta z,\dots,
\zeta z^{n-4}\,\vert\, 1,\dots, z^{n-3}\} \cdot \frac
      {[dz|d\zeta]^3}{\prod_{j=1}^n (z-z_j-\zeta \zeta_j)} \right)$ to
      $- \frac{1}{2^{n-2}} (z_3-z_1 - \tfrac 12 \zeta_3\zeta_1)
      (z_3-z_2 - \tfrac 12 \zeta_3\zeta_2 ) [ dz_4, \dots, dz_n \mid
        d\zeta_3, \dots, d\zeta_n ]$.
\end{prop}

\begin{proof}
We compute the dual \sKS map $\KS^*$ explicitly to be given by the following matrix in the dual bases \eqref{basisA} of $H^0(\omega^3(\NS))$ and $\{ dz_4, \dots, dz_n \mid d\zeta_3, \dots, d\zeta_n \}$ of $H^0(\cMn, \cT_{\cMn}^*)$:
\begin{align*}
     M_n=\begin{pmatrix}
       A & B \\ C & D
     \end{pmatrix}\,,
\end{align*}
where $A$ is an $(n-3)\times (n-3)$ matrix with entries
\begin{align*}
    A_{jk}=\frac {-z_{j+3}^{k-1}}{\prod_{l\neq j+3} (z_{j+3}-z_l)}
\end{align*}
with $1 \le l \le   n$ and $1 \le j,k \le n-3$; $B$ is an $(n-3)\times (n-2)$ matrix with entries
\begin{align*}
    B_{jb}=\frac {z_{j+3}^{b-1}}{\prod_{l\neq j+3} (z_{j+3}-z_l)} \left(-\sum_{i\neq j+3} \frac {\zeta_{j+3}-\zeta_i}{z_{j+3}-z_i}+(b-1)\frac {\zeta_{j+3}}{z_{j+3}}\right)
\end{align*}
with $1 \le l \le   n$ and $1\leq j \leq n-3$, $1\leq b \leq n-2$; $C$ is an $(n-2)\times (n-3)$ matrix with entries
\begin{align*}
    C_{ak}= \frac {\zeta_{a+2} z_{a+2}^{k-1}}{\prod_{l\neq a+2} (z_{a+2}-z_l)}
\end{align*}
with $1 \le l \le   n$ and $1\leq a \leq n-2$, $1\leq k \leq n-3$; and $D$ is an $(n-2)\times (n-2)$ matrix with entries
\begin{align*}
    D_{ab}=-\frac {z_{a+2}^{b-1}}{\prod_{l\neq a+2} (z_{a+2}-z_l)} \left( 2+\sum_{i\neq a+2} \frac {\zeta_{a+2}\zeta_i}{z_{a+2}-z_i} \right)
\end{align*}
with $1 \le l \le   n$ and $1\leq a \leq n-2$, $1\leq b \leq n-2$.
In the end, it is the Berezinian of $M_n$ what we need, and we determine it by a direct computation.
\begin{lem}\label{lem:BerMn}
  The Berezinian of the map $ M_n$ is given by 
\begin{equation*}
    \Ber M_n=\frac{\det(A)}{\det(D-CA^{-1}B)}= - \frac{1}{2^{n-2}} (z_3-z_1 - \tfrac 12 \zeta_3\zeta_1 ) (z_3-z_2 - \tfrac 12 \zeta_3\zeta_2 )
    .
    \label{Ber}
\end{equation*}
\end{lem}

\Cref{app:BerMn} is dedicated to the proof of this lemma. The proposition follows.
\end{proof}

\section{Vertex-operator insertions}
\label{VOI}

Here we discuss vertex-operator insertions and argue that they define a section of the line bundle $B(\omega|_{\NS})$. Our construction in principle follows the lines suggested by Witten \cite[Appendix A]{Witten.2015.NotesOnLowGenus} for the moduli space $M_{g,1}$ of ordinary Riemann surfaces with one puncture, but the case of Neveu-Schwarz punctures on a super Riemann surfaces seems to be considerably subtler. A vertex operator is usually understood as representing the asymptotic state of an incoming (at time~$t=-\infty$), or outgoing (at time~$t=+\infty$) particle, with given spacetime momentum, spin, and other possible quantum numbers. After integration over spacetime, the vacuum expectation value (VEV) of the vertex operator becomes a section of the contangent bundle of the worldsheet (a super Riemann surface) of the propagating superstring, which presents itself as a boundary component of the SRS.
The assumed superconformal invariance of the vertex operator means that the vertex operator is a cotangent vector to the super Riemann surface $X$ at the puncture replacing
the boundary component. This puncture can be of Ramond or Neveu-Schwarz type, but in this paper, we focus on NS punctures.
Let us also consider only \emph{conformal}, i.e., \emph{left-moving}, vertex operators, as a start, to keep doing complex analytic geometry as long as possible. The holomorphic cotangent bundle $\cT^*_{X}$ on $X$ is a rank-$(1|1)$ super vector bundle, which fits into the short exact sequence
\[
0 \to (\cD^*)^{\otimes 2} \to \cT^*_{X} \to \cD^* \to 0,
\]
dual to \eqref{tangentsequence}. Recall that $\cD^* = \omega_X$ and restrict this sequence to the NS puncture $p \in X$:
\[
0 \to \omega_{X}^{\otimes 2}\bigm|_p \to \cT^*_{X}|_p \to \omega_{X}|_p \to 0.
\]
If two cotangent vectors $v_1,v_2 \in \cT^*_{X}|_p$ form a basis such that $v_1$ generates $\omega_{X}^{\otimes 2}|_p$ and $v_2$ projects to a generator of $\omega_{X}|_p$, then their Berezinian $[v_1 | v_2]$ is a vector in $\Ber \cT^*_{X}|_p = \omega_{X}|_p$. This vector is a \emph{vertex-operator insertion}, which is meant to have been integrated over the field configurations (to produce a VEV) and is usually denoted with extra brackets: $\langle [v_1 | v_2] \rangle$.

We claim the vertex-operator insertion $[v_1 | v_2] \in \omega_{X}|_p$ gives rise canonically to a vector in the odd line $\Ber H^0(\NS_p, \omega_{X}|_{\NS_p})$.
Let $\NS_p$ be the NS divisor corresponding to the NS puncture $p$. The NS divisor is determined by its algebra of functions $\cO_X|_{\NS_p} = S(\cD^*|_p) = \cO_X|_p \oplus \cD^*|_p$, see \Cref{sec:foundation}.
Tensoring this algebra with $\omega_X|_{\NS_p}$ over $\cO_X|_{\NS_p}= \cO_{\NS_p}$, we get
\[
\omega_{X}|_{\NS_p} = \omega_{X}|_p \oplus  \omega_{X}^{\otimes 2}|_p,
\]
whence we see that $\Ber H^0(\NS_p, \omega_{X}|_{\NS_p}) = \omega_{X}^{\otimes 2}|_p \otimes \omega^{-1}_{X}|_p = \omega_{X}|_p$. Thus, the vertex-operator insertion $[v_1 | v_2] \in \omega_{X}|_p$ can naturally be regarded as a vector in $B (\omega_{X}|_{\NS_p}) = \Ber H^0(\NS_p, \omega_{X}|_{\NS_p})$. In the above, we have made sure to avoid splitting super vector spaces and sheaves into even and odd components, which allows the argument to easily generalize to families of super Riemann surfaces. For the universal family $\cC_{g,1,0}\to\cM_{g,1,0}$, a vertex-operator insertion becomes a section $V(p) = [v_1 | v_2] \in H^0(\cM_{g,1,0}, B (\omega|_{\NS_p}))$, where $p$ is understood as a section $p: \cM_{g,1,0} \to \cC_{g,1,0}$ of the universal family.

When we have $n$ NS punctures $p_1, \dots p_n$, the corresponding total NS divisor is the sum $\NS = \NS_{p_1} + \dots \linebreak[0] + \NS_{p_n}$,
and the vertex-operator insertions $V_1(p_1), \dots, V_n(p_n)$ multiply to produce a section $\bigotimes_{j=1}^n V_j(p_j)$ of $\Ber(\omega|_{\NS}) = \Ber(\omega|_{\NS_{p_1}}) \otimes \dots \otimes \Ber(\omega|_{\NS_{p_1}})$. If $(z | \zeta)$ are holomorphic coordinates near the $j$'th NS puncture $p_j = (z_j | \zeta_j)$, where $\Big[ \, \zeta [dz | d\zeta]\vert_{\NS_j} \Big| [dz | d\zeta]\vert_{\NS_j}  \Big]$ gives a basis of $\Ber(\omega|_{\NS_j})$, then $V_j(p_j)$ may be expressed as $V_j(z_j | \zeta_j) \Big[ \, \zeta [dz | d\zeta]\vert_{\NS_j} \Big| [dz | d\zeta]\\vert_{\NS_j} \Big]$, whence the product of these gives a section 
\[
\bigotimes_{j=1}^n V_j(z_j | \zeta_j) \Big[ \, \zeta [dz | d\zeta]\vert_{\NS_j} \mid [dz | d\zeta]\vert_{\NS_j}  \Big] \in H^0(\cMn, \Ber (\omega|_{\NS})), 
\]
as claimed in \eqref{VO}, where we used the brackets $\langle \; \rangle$ there to stress that it includes integration over the spacetime configurations.

\section{Tree-level scattering amplitudes in the physics literature}
\label{sec:amplitudes}

To better understand the physics context of our results and match them with the well-known computations, we review how the holomorphic measure $d\nu$
arises
in the determination the superstring scattering amplitudes in the physics literature for the type IIB superstring theory. The matter superfields in terms of the holomorphic coordinates $(z|\zeta)$ on a SRS and the antiholomorphic coordinates $(\bar z|\bar \zeta)$ on the complex conjugate SRS 
are given by
\begin{align}
    {\pmb X}^\mu (z|\zeta;\bar z| \bar \zeta)=X_L^\mu(z)+ X_R^\mu(\bar z)+\zeta \psi^\mu(z) +\bar \zeta \bar \psi^\mu(\bar z).\label{superfield}
\end{align}
In the usual approach to superstring scattering amplitudes one also introduces the superghost fields
\begin{align*}
    B(z|\zeta)=& b(z)+\zeta \beta(z), \qquad C(z|\zeta)= c(z)+\zeta \gamma(z),\\
    \bar B(\bar z|\bar \zeta)=& \bar b(\bar z)+\bar \zeta \bar \beta(\bar z), \qquad \bar C(\bar z|\bar \zeta)= \bar c(\bar z)+\bar \zeta \bar \gamma(\bar z).
\end{align*}
In order to compute the scattering amplitude among $n$ physical states, one inserts and then integrates superconformal vertex operators of conformal weight $(1|1)$ at the punctures.
However, to eliminate the zero modes of the superghosts and get sensible path integrals, one must then also introduce picture changing operators (PCO), necessary to get a sensible nontrivial measure (allowing for nonvanishing result). It turns out that the vertex operators can then be taken in the form
\[
    \cV(z|\zeta;\bar z| \bar \zeta)=c(z)\bar c(\bar z) \delta(\gamma(z)) \delta(\bar \gamma(\bar z)) V(z|\zeta;\bar z| \bar \zeta)\,,
\]
where $V$ is a superconformal field of conformal weight $(1|1)$, and the delta distributions are meant to ensure that the $\gamma$ components vanish at the puncture. When computing the amplitudes, for $n\geq3$, one has also to introduce the insertion of $\delta(\beta) \delta(\bar \beta)$, conventionally chosen at 
$(z_3|\zeta_3,\bar z_3|\bar\zeta_3)$, which can be used to eliminate $\delta(\gamma(z_3)) \delta(\bar \gamma(\bar z_3))$. For $n\geq 4$, $n-3$ vertices can be replaced by integrated vertices
\[
    K_{V_j}=\int_{\Sigma} d\zeta_j d\bar\zeta_j dz_j d\bar z_j\ V_j(z_j|\zeta_j;\bar z_j| \bar \zeta_j)\,,
\]
where $\Sigma$ is the genus zero super Riemann surface underlying the scattering representation, so that we can write for the amplitude
\[
    \cA_n=\left\langle c\bar c \delta(\gamma) V_1(z_1^0|0;\bar z_1^0|0) c\bar c \delta(\gamma) V_2(z_2^0|0;\bar z_2^0|0) 
    c\bar c \delta(\gamma) D_{\zeta_3} D_{\bar \zeta_3} V_3(z_3^0|0;\bar z_3^0|0) \prod_{j=4}^n K_{V_j}\right\rangle\,,
\]
where $\langle\dots \rangle$ means the VEV. In particular, the first three bosonic coordinates are fixed at the arbitrary values $z^0_1, z^0_2, z^0_3$, 
while the first two odd variables are conventionally fixed at $\zeta_a=0$, $a=1,2$.
Notice that the ghosts disappeared from the integrated vertices so they factor in the VEV as
\[
    k_{gh}=\left\langle c\bar c \delta(\gamma) (z_1^0;\bar z_1^0) c\bar c \delta(\gamma) (z_2^0;\bar z_2^0)
    c\bar c \delta(\gamma) (z_3^0;\bar z_3^0) \right\rangle\,.
\]
A simple physics calculation using bosonization and operator product expansion (OPE) then yields
\begin{align*}
   k_{gh}=(z_3^0-z_1^0)(z_3^0-z_2^0)(\bar z_3^0-\bar z_1^0)(\bar z_3^0-\bar z_2^0)\,,
\end{align*}
and we remark that the amplitude can then be rewritten as
\begin{equation}
    \cA_n=g_c^{n-2}\int_{\cMn} \left\langle V_1(z_1^0|0;\bar z_1^0|0) V_2(z_2^0|0;\bar z_2^0|0) V_3(z_3^0|\zeta_3;\bar z_3^0|\bar \zeta_3) 
     \prod_{j=4}^n V_j(z_j|\zeta_j;\bar z_j|\bar \zeta_j) \right\rangle Dm\,, \label{amplitude}
\end{equation}
where $g_c$ denotes the closed string coupling and the integration measure is
\[
    Dm=k_{gh} [dz_4\dots dz_n|d\zeta_3\dots d\zeta_n] [d\bar z_4\dots d\bar z_n|d\bar\zeta_3\dots d\bar\zeta_n] \, .
\]

Direct comparison with \eqref{eq:main} allows us to identify $g_c=1/4$ (when neglecting the dilaton field).
%
%
\section{Conclusion}

We can now compare our results with the standard derivation of tree-level NS amplitudes available in the literature. Instead of deducing the amplitude formula \eqref{amplitude} by constructing a measure through the introduction of ghosts and PCOs, we have suggested the slightly more general formula \eqref{eq:main}, containing only the physical fields and the canonical measure, as determined by the super Mumford isomorphism over $\cMn$. By comparing the two formulas, we see that they coincide after replacing $g_c$ by $\tfrac {\langle e^{\phi}\rangle}{4}$, if we include the dilaton contribution in \eqref{eq:main}.

\appendix

\section{The super Mumford form over $\cMn$
}
\label{app:mapM}
Here we compute the image of the super Mumford form 
\begin{equation}
\label{superMum1}
\mu = \left[\frac{\zeta\vol^3}{z}, \frac{\zeta\vol^3}{z^2}, \frac{\zeta\vol^3}{z^3}  \; \left| \; \frac{\vol^3}{z}, \frac{\vol^3}{z^2} \right. \right] \in B(\omega^3),
\end{equation}
see \eqref{superMum}, tensored with a global section
\[
\left\langle
     \bigotimes_{j=1}^n V_j(z_j | \zeta_j) \Big[ \, \zeta [dz | d\zeta]\\vert_{\NS_j} \Big| [dz | d\zeta]\vert_{\NS_j} \Big]   \right\rangle \in B(\omega|_{\NS})
     \]
arising from vertex-operator insertions, under the isomorphism
\begin{equation}
    \label{phi}
\varphi \co B(\omega^3)  \otimes B(\omega|_{\NS})
\to B(\omega^3(\NS)),
\end{equation}
which comes from the short exact sequence
\begin{equation}
\label{eq:preM1}
 0\to\omega^3\stackrel{f}{\longrightarrow}\omega^3(\NS)\stackrel{g}{\longrightarrow}\omega^3(\NS)|_{\NS}\to 0\,
\end{equation}
and the canonical isomorphism \eqref{core}:
\begin{equation}
\label{core1}
B(\omega|_{\NS}) \xrightarrow{\sim} B(\omega^3(\NS)|_{\NS}) \, .
\end{equation}
Even though all the isomorphisms are defined in a canonical, coordinate-independent way, since we want to perform explicit computations, we work with the standard affine coordinates 
$(z | \zeta)$ on $\PP^{1|1}$.
Note that the isomorphisms are linear over holomorphic functions, such as  $\left\langle \prod_{j=1}^n V_j(z_j | \zeta_j) \right\rangle$, on the supermoduli space and, therefore, it is enough to compute the image of
\begin{equation}
\label{mu-tensored}
\mu \otimes \bigotimes_{j=1}^n \Big[ \, \zeta [dz | d\zeta]\vert_{\NS_j} \mid [dz | d\zeta]\vert_{\NS_j} \Big]
\end{equation}
under $\varphi$.

\begin{prop}
\label{prop}
The image of \eqref{mu-tensored} under the isomorphism $\varphi$ of \eqref{phi} is the Berezinian of the following basis of $H^0(\omega^3(\NS))$:
\begin{align}
\label{basisA}
 \{\zeta,\zeta z,\dots, \zeta z^{n-4}\,\vert\, 1,\dots, z^{n-3}\}  \cdot \frac {[dz|d\zeta]^3}{\prod_{j=1}^n (z-z_j-\zeta \zeta_j)}\,.
\end{align}
\end{prop}

\begin{proof}
For each $j=1, \dots, n$, multiplication by the generator $\frac {[dz|d\zeta]^2}{ (z-z_j-\zeta \zeta_j)}$ of $\omega^2 (NS)$ (in a small enough open neighborhood of the $j$th NS puncture) maps the basis
\begin{equation*}
\left\{ \zeta \vol \biggm|_{\NS_j} \;
\Biggm| \; \vol \biggm|_{\NS_j}
\right\}
\end{equation*}
of $H^0(\omega\vert_{\NS_j})$ to the basis
\begin{equation} 
\label{NS|NS0}
\left\{ \frac {\zeta \vol^3 }{z-z_j-\zeta \zeta_j} \biggm|_{\NS_j} \;
\Biggm| \; \frac {\vol^3}{z-z_j-\zeta \zeta_j} \biggm|_{\NS_j}
\right\}
\end{equation}
of $H^0(\omega^3(\NS)\vert_{\NS_j})$,
where the ``restriction'' $\vert_{\NS_j}$ means considering a local section of $\omega^3(\NS)$ modulo the subsheaf $(z-z_j-\zeta \zeta_j)\, \omega^3(\NS)$. This maps the Berezinian 
\begin{equation}
\label{|NS}
\bigotimes_{j=1}^n \Big[ \, \zeta [dz | d\zeta]\vert_{\NS_j} \Big| [dz | d\zeta]\vert_{\NS_j} \Big]
\end{equation}
of the basis of $H^0(\omega_{\NS})$ to the Berezinian
\begin{equation}
\label{NS|NS}
\bigotimes_{j=1}^n \left[ \frac {\zeta \vol^3 }{z-z_j-\zeta \zeta_j} \biggm|_{\NS_j} \;
\Biggm| \; \frac {\vol^3}{z-z_j-\zeta \zeta_j} \biggm|_{\NS_j}
\right]
\end{equation}
of the basis of $H^0(\omega^3(\NS)_{\NS})$
and induces the canonical isomorphism \eqref{core1} on the Berezinians of the cohomology.
It remains to show that the canonical isomorphism
\begin{equation}
 \label{eq:preomega3NS}
 B (\omega^3)  \otimes  B (\omega^3(\NS)|_{\NS}) \xrightarrow{\sim}B(\omega^3(\NS))
\end{equation}
induced by \eqref{eq:preM1} matches the super Mumford form $\mu$, the Berezinian \eqref{NS|NS}, and the Berezinian of the basis \eqref{basisA}:
\begin{equation}
\label{matching}
\mu \otimes \eqref{NS|NS} \mapsto [\eqref{basisA}] .
\end{equation}


We will work with the long exact sequence in cohomology
\begin{equation}
\label{eq:M1}
    0\rightarrow H^0(\omega^3(\NS)) \stackrel {g}{\longrightarrow} H^0(\omega^3(\NS)\vert_{\NS}) \stackrel {\delta}{\longrightarrow} H^1(\omega^3)\rightarrow 0\,
\end{equation}
associated with the short exact sequence \eqref{eq:preM1} of sheaves. In terms of the sequence \eqref{eq:M1}, the isomorphism \eqref{eq:preomega3NS} rewrites as
\[
\Ber (H^0(\omega^3(\NS)|_{\NS})) \xrightarrow{\sim}  \Ber (H^1(\omega^3)) \otimes   \Ber (H^0(\omega^3(\NS)))   .
\]
We will describe the relation between the bases of the three terms of \eqref{eq:M1} and the resulting relation between their Berezinians.

The map~$g$ in \eqref{eq:M1} is obtained by restricting the elements of the basis \eqref{basisA} to the divisor $\NS$. This gives the following elements of  $H^0(\omega^3(\NS)|_{\NS})$:
\begin{gather*}
    \bigoplus_{j} \frac {\zeta z_j^a}{\prod_{k\neq j} (z_j-z_k)}\,  \frac {\vol^3}{z-z_j-\zeta \zeta_j} \biggm|_{\NS_j} , \qquad 0 \le a \le n-4,\\
 \bigoplus_{j} \left[ \frac {z_j^a}{\prod_{k\neq j}(z_j-z_k)} - \frac {\zeta z_j^a}{\prod_{k\neq j}(z_j-z_k)} \left( \sum_{k\neq j} \frac {\zeta_j - \zeta_k}{z_j-z_k} - a\frac {\zeta_j}{z_j} \right) \right] \frac {\vol^3}{z-z_j-\zeta \zeta_j} \biggm|_{\NS_j}, \quad 0 \le a \le n-3.
\end{gather*}
To determine the map~$\delta$ in \eqref{eq:M1} explicitly, we first note that
\begin{equation}
\label{H1omega3}
 \left\{\zeta z^2 , \zeta z, \zeta \biggm| z^2, z  \right\} \cdot \frac{\vol^3}{z^3}
\end{equation}
is a basis of $H^1(\omega^3)$ whose Berezinian is the super Mumford form $\mu$, see \eqref{superMum1}.
On the other hand, by Serre duality $H^1(\omega^3) = H^0(\omega^{-2})^*$, and
\[
\left\{ 1, z, z^2\, \vert\, -\zeta, -\zeta z \right\} \cdot \vol^{-2} 
\]
is the dual basis of $H^0(\omega^{-2})$.
The map $\delta$
sends $\lambda \in H^0(\omega^3(\NS)|_{\NS}) $ to $\delta(\lambda) \in H^1(\omega^3) = H^0(\omega^{-2})^* $ such that for any $\alpha \in H^0(\omega^{-2})$, we have
\begin{align*}
    \delta(\lambda)(\alpha)=\Res_{\NS} (\lambda \alpha).
\end{align*}
This map $\delta$ is given explicitly on the basis \eqref{NS|NS0} by
\begin{align*}
    \frac {\zeta  \vol^3 }{z-z_j-\zeta \zeta_j}\biggm|_{\NS_j}&\longmapsto \left( \frac \zeta{z}+z_j \frac \zeta{z^2}+z_j^2 \frac \zeta{z^2} \right) [dz|d\zeta]^3 , \\
    \frac { \vol^3}{z-z_j-\zeta \zeta_j}\biggm|_{\NS_j} &\longmapsto  \left( \frac 1{z}+z_j \frac 1{z^2} \right) [dz|d\zeta]^3
\end{align*}
for $j=1,\dots,n$.

Thus, in matrix form with respect to the chosen bases, we can write
\begin{align*}
    g=\begin{pmatrix}
        A & B \\ 0 & C
    \end{pmatrix},
\end{align*}
with
\begin{align*}
    A=&\begin{pmatrix}
      \frac 1{\prod_{k\neq 1} (z_1-z_k)} (1,z_1,\ldots,z_1^{n-4}) \\
      \frac 1{\prod_{k\neq 2} (z_2-z_k)} (1,z_2,\ldots,z_2^{n-4}) \\
      \ldots\ldots\ldots \\
      \frac 1{\prod_{k\neq n} (z_n-z_k)} (1,z_n,\ldots,z_n^{n-4})
    \end{pmatrix}, \\
    B=&\begin{pmatrix}
      \frac 1{\prod_{k\neq 1} (z_1-z_k)} \left(\sum_{k\neq 1} \frac {\zeta_1 - \zeta_k}{z_1-z_k}\ ;\ 
       \sum_{k\neq 1} \frac {(\zeta_1 - \zeta_k) z_1}{z_1-z_k} - \zeta_1\ ;\ \ldots\ ;\ \sum_{k\neq 1} \frac {(\zeta_1 -\zeta_k) z_1^{n-3}}{z_1-z_k} - (n-3)z_1^{n-4} \zeta_1\right) \\
      \frac 1{\prod_{k\neq 2} (z_2-z_k)} \left(\sum_{k\neq 2} \frac{\zeta_2 - \zeta_k}{z_2-z_k}\ ;\ 
       \sum_{k\neq 2} \frac {(\zeta_2 - \zeta_k) z_2}{z_2-z_k}- \zeta_2\ ;\ \ldots\ ;\ \sum_{k\neq 2} \frac {(\zeta_2- \zeta_k) z_2^{n-3}}{z_2-z_k} - (n-3)z_2^{n-4} \zeta_2\right) \\
      \ldots\ldots\ldots \\
      \frac 1{\prod_{k\neq n} (z_n-z_k)} \left(\sum_{k\neq n} \frac {\zeta_n - \zeta_k}{z_n-z_k}\ ;\ 
      \sum_{k\neq n} \frac {( \zeta_n - \zeta_k) z_n}{z_n-z_k} -  \zeta_n\ ;\ \ldots\ ;\ \sum_{k\neq n} \frac {(  \zeta_n - \zeta_k) z_n^{n-3}}{z_n-z_k} - (n-3)z_n^{n-4} \zeta_n \right)
    \end{pmatrix}, \\
    C=&\begin{pmatrix}
      \frac 1{\prod_{k\neq 1} (z_1-z_k)} (1,z_1,\ldots,z_1^{n-3}) \\
      \frac 1{\prod_{k\neq 2} (z_2-z_k)} (1,z_2,\ldots,z_2^{n-3}) \\
      \ldots\ldots\ldots \\
      \frac 1{\prod_{k\neq n} (z_n-z_k)} (1,z_n,\ldots,z_n^{n-3})
    \end{pmatrix},
\end{align*}
and 
\begin{align*}
    \delta=\begin{pmatrix}
        1 & 1 & \ldots & 1 & 0 & 0 & \ldots & 0 \\
        z_1 & z_2 & \ldots & z_n & 0 & 0 & \ldots & 0 \\
        z_1^2 & z_2^2 & \ldots & z_n^2 & 0 & 0 & \ldots & 0 \\
        0 & 0 & \ldots & 0 & 1 & 1 & \ldots & 1 \\
        0 & 0 & \ldots & 0 & z_1 & z_2 & \ldots & z_n
    \end{pmatrix}.
\end{align*}

Let us pick a splitting of the short exact sequence \eqref{eq:M1}, that is to say, a map $\delta': H^1(\omega^3) \to H^0(\omega^3(\NS)\vert_{\NS})$ such that $\delta \delta'= \id$. A priori there are many choices, but the relation between the Berezinians will not depend on the choice. Since we have fixed the coordinates $\{z_4,\ldots z_n\,|\,\zeta_3,\ldots,\zeta_n\}$,
we are left to work with the matrix elements associated to the remaining ones, which renders
an invertible $5\times 5$ matrix. We get
\begin{align*}
    \delta'=\begin{pmatrix}
        \frac {z_2 z_3}{(z_2-z_1)(z_3-z_1)} & \frac {-(z_2+ z_3)}{(z_2-z_1)(z_3-z_1)} & \frac {1}{(z_2-z_1)(z_3-z_1)} & 0 & 0 \\
        \frac {z_1 z_3}{(z_1-z_2)(z_3-z_2)} & \frac {-(z_1+ z_3)}{(z_1-z_2)(z_3-z_2)} & \frac {1}{(z_1-z_2)(z_3-z_2)} & 0 & 0 \\
        \frac {z_2 z_1}{(z_2-z_3)(z_1-z_3)} & \frac {-(z_2+ z_1)}{(z_2-z_3)(z_1-z_3)} & \frac {1}{(z_2-z_3)(z_1-z_3)} & 0 & 0 \\
        0 & 0 & 0 & 0 & 0 \\
        \vdots & \vdots & \vdots & \vdots & \vdots \\
        0 & 0 & 0 & 0 & 0 \\
        0 & 0 & 0 & \frac {z_2}{z_2-z_1} &\frac{-1}{z_2-z_1} \\
        0 & 0 & 0 & \frac {z_1}{z_1-z_2} &\frac{-1}{z_1-z_2} \\
        0 & 0 & 0 & 0 & 0 \\
        \vdots & \vdots & \vdots & \vdots & \vdots \\
        0 & 0 & 0 & 0 & 0 
    \end{pmatrix}.
\end{align*}
The transition matrix $M$ from the basis of $H^0(\omega^3(\NS)\vert_{\NS})$ formed by the image of the basis \eqref{basisA} under $g$ and the image of the basis \eqref{H1omega3} under $\delta'$ to the basis \eqref{NS|NS0} over all $j = 1, \dots, n$ is thus represented by a $2n\times 2n$ block matrix
\begin{align*}
    M=\begin{pmatrix}
        \tilde A & \tilde B \\ 0 & \tilde C
    \end{pmatrix},
\end{align*}
where
\begin{align*}
    \tilde A=&\begin{pmatrix}
      \frac 1{\prod_{k\neq 1} (z_1-z_k)} (1,z_1,\ldots,z_1^{n-4}), & \frac {(z_2z_3, -z_2-z_3, 1)}{(z_2-z_1)(z_3-z_1)} \\
      \frac 1{\prod_{k\neq 2} (z_2-z_k)} (1,z_2,\ldots,z_2^{n-4}), & \frac {(z_1z_3, -z_1-z_3, 1)}{(z_1-z_2)(z_3-z_2)} \\
      \frac 1{\prod_{k\neq 3} (z_3-z_k)} (1,z_3,\ldots,z_3^{n-4}), & \frac {(z_1z_2, -z_1-z_2, 1)}{(z_1-z_3)(z_2-z_3)} \\
      \frac 1{\prod_{k\neq 4} (z_4-z_k)} (1,z_4,\ldots,z_4^{n-4}), & (0, 0, 0) \\
      \vdots & \vdots \\
      \frac 1{\prod_{k\neq n} (z_n-z_k)} (1,z_n,\ldots,z_n^{n-4}), & (0, 0, 0)
    \end{pmatrix}, \\
    \tilde B=&\begin{pmatrix}
      \frac 1{\prod_{k\neq 1} (z_1-z_k)} \left(\sum_{k\neq 1} \frac {\zeta_1 - \zeta_k}{z_1-z_k};\ 
       \sum_{k\neq 1} \frac {(\zeta_1 - \zeta_k) z_1}{z_1-z_k} - \zeta_1;\ \ldots ;\ \sum_{k\neq 1} \frac {(\zeta_1 -\zeta_k) z_1^{n-3}}{z_1-z_k} - (n-3)z_1^{n-4} \zeta_1;0 ; 0 \right) \\
      \frac 1{\prod_{k\neq 2} (z_2-z_k)} \left(\sum_{k\neq 2} \frac{\zeta_2 - \zeta_k}{z_2-z_k};\ 
       \sum_{k\neq 2} \frac {(\zeta_2 - \zeta_k) z_2}{z_2-z_k}- \zeta_2;\ \ldots;\ \sum_{k\neq 2} \frac {(\zeta_2- \zeta_k) z_2^{n-3}}{z_2-z_k} - (n-3)z_2^{n-4} \zeta_2;0; 0 \right) \\
      \ldots\ldots\ldots \\
      \frac 1{\prod_{k\neq n} (z_n-z_k)} \left(\sum_{k\neq n} \frac {\zeta_n - \zeta_k}{z_n-z_k};\ 
      \sum_{k\neq n} \frac {( \zeta_n - \zeta_k) z_n}{z_n-z_k} -  \zeta_n;\ \ldots;\ \sum_{k\neq n} \frac {(  \zeta_n - \zeta_k) z_n^{n-3}}{z_n-z_k} - (n-3)z_n^{n-4} \zeta_n;0;0 \right)
    \end{pmatrix}, \\
    \tilde C=&\begin{pmatrix}
      \frac 1{\prod_{k\neq 1} (z_1-z_k)} (1,z_1,\ldots,z_1^{n-3}), & \frac {(z_2, -1)}{z_2-z_1} \\
      \frac 1{\prod_{k\neq 2} (z_2-z_k)} (1,z_2,\ldots,z_2^{n-3}), & \frac {(z_1, -1)}{z_1-z_2} \\
      \frac 1{\prod_{k\neq 3} (z_3-z_k)} (1,z_3,\ldots,z_3^{n-3}), & (0, 0, 0) \\
      \vdots & \vdots \\
      \frac 1{\prod_{k\neq n} (z_n-z_k)} (1,z_n,\ldots,z_n^{n-3}), & (0, 0, 0)
    \end{pmatrix}.
\end{align*}
Accordingly, the Berezinians of the two bases  of $H^0(\omega^3(\NS)|_{\NS})$ are related by $\Ber M$, which the following lemma claims to be simply equal to 1, proving \eqref{matching} and thereby the proposition.
\end{proof}

\begin{lem}
 $\Ber M=1$.
\end{lem}

\begin{proof}
We have $\Ber M=\det \tilde A \det \tilde C^{-1}$. In order to compute these determinants, let us notice that
\begin{align*}
\tilde A=\begin{pmatrix}
        A_1 & A_2 \\ A_3 & 0
    \end{pmatrix},    
\end{align*}
where $A_2$ is the $3\times 3$ matrix
\begin{align*}
    A_2=\begin{pmatrix}
        1 & 1 & 1\\
        z_1 & z_2 & z_3\\
        z_1^2 & z_2^2 & z_3^2
    \end{pmatrix}^{-1},
\end{align*}
and $A_3$ is the $(n-3)\times (n-3)$ matrix
\begin{align*}
    A_3=\begin{pmatrix}
    \frac 1{\prod_{k\neq 4} (z_4-z_k)} (1,z_4,\ldots,z_4^{n-4}) \\
      \vdots \\
      \frac 1{\prod_{k\neq n} (z_n-z_k)} (1,z_n,\ldots,z_n^{n-4})
    \end{pmatrix}.
\end{align*}
Therefore,
\begin{align*}
    \det A_2=&\prod_{1\leq j<k \leq 3}\frac 1{z_k-z_j},\\
    \det A_3=&\prod_{l=4}^n \left( \frac 1{\prod_{k\neq l} (z_l-z_k)} \right)\det \begin{pmatrix}
    1,z_4,\ldots,z_4^{n-4} \\
      \vdots \\
     1,z_n,\ldots,z_n^{n-4}
    \end{pmatrix}=\prod_{l=4}^n \left( \frac 1{\prod_{k\neq l} (z_l-z_k)} \right) \prod_{4\leq r< m\leq n}(z_m-z_r),
\end{align*}
and\footnote{In bringing each of the far right columns to the left, in order to get a block-diagonal matrix, the determinant changes sign $3(n-3)$ times.}
\begin{align*}
    \det \tilde A=(-1)^{n-3} \det A_2 \det A_3=(-1)^{n-3}\prod_{1\leq j<k \leq 3}\frac 1{z_k-z_j} \prod_{l=4}^n \left( \frac 1{\prod_{k\neq l} (z_l-z_k)} \right) \prod_{4\leq r< m\leq n}(z_m-z_r).
\end{align*}
In the same way, we have
\begin{align*}
\tilde C=\begin{pmatrix}
        C_1 & C_2 \\ C_3 & 0
    \end{pmatrix},    
\end{align*}
where $C_2$ is the $2\times 2$ matrix
\begin{align*}
    C_2=\begin{pmatrix}
        1 & 1 \\
        z_1 & z_2
    \end{pmatrix}^{-1},
\end{align*}
and $C_3$ is the $(n-2)\times (n-2)$ matrix
\begin{align*}
    C_3=\begin{pmatrix}
    \frac 1{\prod_{k\neq 3} (z_3-z_k)} (1,z_3,\ldots,z_3^{n-3}) \\
      \vdots \\
      \frac 1{\prod_{k\neq n} (z_n-z_k)} (1,z_n,\ldots,z_n^{n-3})
    \end{pmatrix}.
\end{align*}
Therefore,
\begin{align*}
    \det C_2=& \frac 1{z_2-z_1},\\
    \det C_3=&\prod_{l=3}^n \left( \frac 1{\prod_{k\neq l} (z_l-z_k)} \right)\det \begin{pmatrix}
    1,z_3,\ldots,z_3^{n-3} \\
      \vdots \\
     1,z_n,\ldots,z_n^{n-3}
    \end{pmatrix}=\prod_{l=3}^n \left( \frac 1{\prod_{k\neq l} (z_l-z_k)} \right) \prod_{3\leq r< m\leq n}(z_m-z_r), 
\end{align*}
and 
\begin{align*}
    \det \tilde C= \det C_2 \det C_3=\frac 1{z_2-z_1} \prod_{l=3}^n \left( \frac 1{\prod_{k\neq l} (z_l-z_k)} \right) \prod_{3\leq r< m\leq n}(z_m-z_r).
\end{align*}
By taking the quotient we get
\begin{align*}
    \Ber M=& \frac {\det \tilde A}{\det \tilde C}=\frac {(-1)^{n-3}}{(z_3-z_2)(z_3-z_1)(z_2-z_1)} \prod_{l=4}^n \left( \frac 1{\prod_{k\neq l} (z_l-z_k)} \right) \prod_{4\leq r< m\leq n}(z_m-z_r)\cr
    &\times (z_2-z_1) \prod_{l=3}^n \left(\prod_{k\neq l} (z_l-z_k) \right) \prod_{3\leq r< m\leq n}\frac 1{z_m-z_r}\cr
    =& \frac {(-1)^{n-3}}{(z_3-z_2)(z_3-z_1)} \left(\prod_{k\neq 3} (z_3-z_k)\right) \prod_{3< m\leq n}\frac 1{z_m-z_3}.
\end{align*}
The first factor times the last $n-3$ give $\prod_{k\neq 3} \frac 1{z_3-z_k}$, so we get $\Ber M=1$.
\end{proof}

\section{Proof of \Cref{lem:BerMn}: the Berezinian of the \sKS map}
\label{app:BerMn}

We first notice that, for $N=D-CA^{-1}B$, in terms of rows, we can
write
\begin{align*}
    N=\begin{pmatrix}
        -\frac 2{\prod_{j\neq 3}(z_3-z_j)} \left( \left(1+\frac 12 \sum_{k\neq 3} \frac {\zeta_3\zeta_k}{z_3-z_k}\right)
        (1,z_3,\dots,z_3^{n-3})+\frac {\zeta_3}2 (1,z_3,\dots,z_3^{n-4}) A^{-1}B \right)\\
        -\frac 2{\prod_{j\neq 4}(z_4-z_j)} \left( 1,z_4,\ldots,z_4^{n-3} \right)\\
        \ldots\\
        -\frac 2{\prod_{j\neq n}(z_n-z_j)} \left( 1,z_n,\ldots,z_n^{n-3} \right)
    \end{pmatrix}
\end{align*}
so that writing the first row as a sum, we can write $\det N=\det N' + \det N''$, where
\begin{align*}
    N'=\begin{pmatrix}
        -\frac 2{\prod_{j\neq 3}(z_3-z_j)} \left( \left(1+\frac 12 \sum_{k=1}^2 \frac {\zeta_3\zeta_k}{z_3-z_k}\right)
        (1,z_3,\dots,z_3^{n-3}) \right)\\
        -\frac 2{\prod_{j\neq 4}(z_4-z_j)} \left( 1,z_4,\ldots,z_4^{n-3} \right)\\
        \ldots\\
        -\frac 2{\prod_{j\neq n}(z_n-z_j)} \left( 1,z_n,\ldots,z_n^{n-3} \right)
    \end{pmatrix}
\end{align*}
and
\begin{align*}
    N''=\begin{pmatrix}
        -\frac 1{\prod_{j\neq 3}(z_3-z_j)} \left( \left( \sum_{k=4}^n \frac {\zeta_3\zeta_k}{z_3-z_k}\right)
        (1,z_3,\dots,z_3^{n-3})+\zeta_3 (1,z_3,\dots,z_3^{n-4}) A^{-1}B \right)\\
        -\frac 2{\prod_{j\neq 4}(z_4-z_j)} \left(1,z_4,\ldots,z_4^{n-3} \right)\\
        \ldots\\
        -\frac 2{\prod_{j\neq n}(z_n-z_j)} \left( 1,z_n,\ldots,z_n^{n-3} \right)
    \end{pmatrix}.
\end{align*}
It is immediate to see that $\det A/\det N'$ gives the expression \eqref{Ber}. So we just need to prove that $\det N''=0$. This is equivalent to say that
$\det \tilde N=0$, where
\begin{align*}
    \tilde N=\begin{pmatrix}
        \left( \sum_{k=4}^n \frac {\zeta_3\zeta_k}{z_3-z_k}\right)
        (1,z_3,\dots,z_3^{n-3})+\zeta_3 (1,z_3,\dots,z_3^{n-4}) A^{-1}B \\
        (1, z_4,\ldots,z_4^{n-3}) \\
        \ldots\\
        (1,z_n,\ldots,z_n^{n-3} )
    \end{pmatrix}.
\end{align*}
Now, we observe that the first line has the form $\sum_{k=4}^n \chi_k \zeta_3\zeta_k$, where $\chi_k$ is a row. So, it is sufficient to prove that
$\det \tilde N_k=0$ for each $k$, where $\tilde N_k$ is obtained by replacing $\chi_k$ in place of the first row of $\tilde N$. By the structure of the element of the matrix $\tilde N$, we see that its determinant is cyclically symmetric in the $z_k$, $k>3$, so it is sufficient to prove the vanishing result just for one given $k$, say $k=4$. We also notice that in $A^{-1}B$ the common denominators in $A$ and $B$ cancel out. It is also convenient to notice that of the rows of $B$ only the last terms in the parenthesis of $B_{jb}$ contribute to the determinant, since the first terms give rows that are combinations of the remaining rows. Therefore, we can replace the matrix $A^{-1}B$ with $\tilde A^{-1}\tilde B_4$, where
\begin{align*}
    \tilde A = \begin{pmatrix}
        1 & z_4 & z_4^2&\ldots & z_4^{n-4} \\
        1 & z_5 & z_5^2&\ldots & z_5^{n-4} \\
        \vdots&\vdots&\vdots & \ddots & \vdots\\
        1 & z_n & z_n^2&\ldots & z_n^{n-4} \\
    \end{pmatrix}
\end{align*}
is a Vandermonde matrix (up to the change of sign in the second column), and
\begin{align*}
    \tilde B_4= \begin{pmatrix}
        0 & -\zeta_4 & -2\zeta_4 z_4&\ldots & -(n-3) \zeta_4 z_4^{n-4} \\
        0&0&0 & \ldots & 0 \\
        \vdots&\vdots&\vdots & \ddots & \vdots\\
        0&0&0 & \ldots & 0
    \end{pmatrix}.
\end{align*}
Thus, we see that
\begin{align*}
    \det \tilde N_4=\det \tilde N_4'+\det \tilde N_4''\,,
\end{align*}
where
\begin{align*}
    \tilde N_4'=\begin{pmatrix}
        \frac {\zeta_3 \zeta_4}{z_3-z_4}
        (1,z_3,\dots,z_3^{n-3}) \\
        1,z_4,\ldots,z_4^{n-3} \\
        \ldots\\
        1,z_n,\ldots,z_n^{n-3} 
    \end{pmatrix}
\end{align*}
and
\begin{align*}
    \tilde N_4''=\begin{pmatrix}
        \zeta_3 (1,z_3,\dots,z_3^{n-4}) \tilde A^{-1}\tilde B_4 \\
        1,z_4,\ldots,z_4^{n-3} \\
        \ldots\\
        1,z_n,\ldots,z_n^{n-3} 
    \end{pmatrix}.
\end{align*}
We now prove that $\det \tilde N_4''=-\det \tilde N_4'$. Factoring out $\tfrac{\zeta_3 \zeta_4}{z_3-z_4}$ from the first row of $\tilde N_4'$ gives a Vandermonde matrix, so that the determinant is equal to
\begin{align}
    \det \tilde N_4'=\frac {\zeta_3 \zeta_4}{z_3-z_4} \prod_{3\leq j<k\leq n} (z_k-z_j)=-\zeta_3 \zeta_4 \prod_{l=5}^n (z_l-z_3) \prod_{4\leq j<k\leq n} (z_k-z_j).
    \label{det N'}
\end{align}
We can now compute $\det \tilde N_4''$. To this end, we notice that the entry number $k$ of the first row of $\tilde N_4''$ is $\sum_{a=1}^{n-3} z_3^{a-1} \tilde A^{-1}_{a1}(\tilde B_4)_{1k}$. This can be computed as follows. Since $\tilde A^{-1}_{a1}\det \tilde A $ is the co-factor of the matrix $\tilde A$ in position $a1$, we see that
$\sum_{a=1}^{n-3} z_3^{a-1} \tilde A^{-1}_{a1}\det \tilde A$ is the determinant of the matrix $\bar A$ obtained by replacing $z_4$ with $z_3$ in the first row of $\tilde A$. Since both $\tilde A$ and $\tilde A'$ are Vandermonde matrices, we get
\begin{align*}
    \sum_{a=1}^{n-3} z_3^{a-1} \tilde A^{-1}_{a1}=\frac {\det \tilde A'}{\det \tilde A}=\prod_{j=5}^n \frac {z_j-z_3}{z_j-z_4},
\end{align*}
and so
\begin{align}
    \det \tilde N_4''=&-\zeta_3 \zeta_4 \frac {\prod_{l=5}^n (z_l-z_3)}{\prod_{j=5}^n (z_j-z_4)} \det
    \begin{pmatrix}
        (0,1,2z_4,\dots,(n-3) z_4^{n-4}) \\
        (1,z_4,\ldots,z_4^{n-3}) \\
        \ldots\\
        (1,z_n,\ldots,z_n^{n-3} )
    \end{pmatrix}. \label{det N''}
\end{align}
Now, all co-factors of the first line of the last matrix are of Vandermonde type and, in particular, are thus divisible by $\prod_{j=5}^n (z_j-z_4)$. It follows that $\det \tilde N_4''$ is a polynomial in the $z_j$ of the same degree as $\det \tilde N_4'$, which vanishes when $z_j=z_3$ (for some $j\ge 5$) or when $z_j=z_k$ (for some $5\le j<k\le n$). Since the determinant $\det \tilde N_4'$ is simply the product of these linear factors, and $\det\tilde N_4''$ is a polynomial of the same degree divisible by the same degree, it follows that $\det \tilde N_4''/\det \tilde N_4'$ is just a constant, which we can compute by evaluating for any specific values of $z_j$. For $z_4=0$ we get
\begin{align}
    \det \tilde N_4''=&-\zeta_3 \zeta_4 \frac {\prod_{l=5}^n (z_l-z_3)}{\prod_{j=5}^n z_j} \det
    \begin{pmatrix}
        (0,1,0,\dots,0) \\
        (1,0,\ldots,0) \\
        (1,z_5,\ldots,z_5^{n-3} )\\
        \ldots\\
        (1,z_n,\ldots,z_n^{n-3} )
    \end{pmatrix}=\zeta_3 \zeta_4 \frac {\prod_{l=5}^n (z_l-z_3)}{\prod_{j=5}^n z_j} \det
    \begin{pmatrix}
        (z_n^2,\ldots,z_5^{n-3} )\\
        \ldots\\
        (z_n^2,\ldots,z_n^{n-3} )
    \end{pmatrix}\cr
    =&\zeta_3 \zeta_4 \prod_{l=5}^n (z_l-z_3)\prod_{j=5}^n z_j \prod_{5\leq j<k \leq n} (z_k-z_j),
\end{align}
which coincides with the evaluation of $-\det \tilde N_4'$, computed in \eqref{det N'}, for $z_4=0$.

\section{The $z_1 = \infty$ case}
\label{infty}

Here we consider the  $z_1=\infty$ case, or, equivalently, of $(w \, \vert \, \eta) = (0
\, \vert \, \eta_1)$ in the coordinates $(w \, \vert \, \eta) = \left(\tfrac1z\,, \vert \, \tfrac{\zeta}{z}\right)$ in the chart on $\PP^{1|1}$ centered at $z = \infty$. We
will deduce the resulting expression \eqref{Ber3} in that case of \Cref{1} by letting $z_1 \to \infty$ in formula \eqref{main}, which
gives the expression of the measure for a finite value of $z_1$. More precisely, we will let $(w_1 \, | \, \eta_1)
\to (0 \, | \, \eta_1)$ (or simply $w_1 \to 0$) for $(w_1 \, \vert \,
\eta_1) = \left(\tfrac{1}{z_1} \, \vert \, \tfrac{\zeta_1}{z_1}\right)$.


There is one subtlety, though. We  want
the vertex-operator insertion 
\begin{equation}
  \label{limit}
V_1 (z_1 | \zeta_1) \left[ \, \zeta \vol \! \bigm|_{\NS_1} 
  \Bigm| \, \vol \! \bigm|_{\NS_1} \, 
  \right] = V_1 \left( \tfrac{1}{w_1} \Big| \tfrac{\zeta_1}{w_1} \right) \tfrac{1}{w_1}\left[ \eta [dw | d\eta]
   \! \bigm|_{\NS_1} \; \Bigm| \; [dw |
      d\eta] \! \bigm|_{\NS_1} \right]
\end{equation}
at $\NS_1 = \{w - w_1 + \eta \eta_1 = 0\}$ in \eqref{VO} to converge to a finite vertex operator
\[
\widetilde{V_1} (0|\eta_1) \Big[ \, \eta [dw | d\eta]\vert_{\NS_1} \Big| [dw | d\eta]\vert_{\NS_1} \Big]
\]
as $w_1 \to 0$, i.e., we want
\[
\lim_{w_1 \to 0} V_1 \left( \tfrac{1}{w_1} \Big| \tfrac{\zeta_1}{w_1} \right) \tfrac{1}{w_1} = \widetilde{V_1} (0|\eta_1) \, .
\]
Here we again drop the vacuum expectation value brackets $\langle \, \rangle$ from the notation.

On the other hand,
\begin{equation*}
d \nu = - \frac{1}{2^{n-2} w_1} \left(z_3 w_1 - 1 -\tfrac 12
\zeta_3\eta_1\right) \left(z_3-z_2 - \tfrac 12
\zeta_3\zeta_2\right) [dz_4\dots dz_n\,|\,d\zeta_3\dots d\zeta_n]
\end{equation*}
and hence
\begin{align*}
V_1 (z_1 | \zeta_1) d \nu & = - \frac{V_1 \big( \frac{1}{w_1} \mid \frac{\zeta_1}{w_1} \big)}{2^{n-2} w_1} \left(z_3 w_1 - 1 -\tfrac 12
\zeta_3\eta_1\right) \left(z_3-z_2 - \tfrac 12
\zeta_3\zeta_2\right) [dz_4\dots dz_n\,|\,d\zeta_3\dots d\zeta_n]\\
& \to \quad \widetilde{V_1} (0|\eta_1) \tfrac{1}{2^{n-2}} \left(1 + \tfrac 12
\zeta_3\eta_1\right) \left(z_3-z_2 - \tfrac 12
\zeta_3\zeta_2\right) [dz_4\dots dz_n\,|\,d\zeta_3\dots d\zeta_n]
\end{align*}
as $w_1 \to 0$. In
\Cref{1}, we use $V_1$ in lieu of $\widetilde{V_1}$ as well as slightly abuse notation and denote the above limit by $V_1(0 | \eta_1) d\nu$.

\end{document}